\documentclass[11pt]{article}

\usepackage{amssymb} 
\usepackage{amsmath}     
\usepackage{amsthm}
\usepackage{todonotes}
\usepackage{mathtools}
\usepackage{algorithm, algorithmic}

\usepackage{a4wide}
\usepackage{graphicx}
\usepackage{hyperref}
\newcommand{\Rset}{\mathbb{R}}

\newtheorem{lem}{Lemma}

\begin{document}

\title{Robust recoverable 0-1 optimization problems under polyhedral uncertainty} 

\author{Mikita Hradovich$^\dag$, Adam Kasperski$^\ddag$,  Pawe{\l} Zieli{\'n}ski$^\dag$\\
          {\small \textit{$^\dag$Faculty of Fundamental Problems of Technology,}}\\
	{\small \textit{Wroc{\l}aw University of Technology,  Wroc{\l}aw, Poland}}\\
	 {\small \textit{$^\ddag$Faculty of Computer Science and Management,}}\\
	{\small \textit{Wroc{\l}aw University of Technology,  Wroc{\l}aw, Poland}}\\
	{\small \texttt{\{mikita.hradovich,adam.kasperski,pawel.zielinski\}@pwr.edu.pl}}
}
 \date{}
    
\maketitle

 \begin{abstract}
 This paper deals with
 a robust recoverable approach to  0-1 programming problems. It is assumed that a solution constructed in the first stage can be modified to some extent in the second stage. This modification consists in choosing a solution in some prescribed neighborhood of the current solution. The second stage solution cost can be uncertain and a polyhedral structure of uncertainty is used. The resulting robust recoverable problem is a min-max-min problem, which can be hard to solve when the number of variables is large. 
In this paper we provide  a framework for solving robust recoverable 0-1 programming problems with a specified polyhedral uncertainty and
propose several lower bounds and approximate solutions, which can be used for a wide class of 0-1 optimization problems. The results of computational tests for two problems, namely the assignment and the knapsack ones, are also presented.
 \end{abstract}

\section{Introduction}

In many practical applications of discrete optimization a solution constructed in the first stage can be modified to some extent, after the problem parameters have been changed. For example, a connection network (represented by a spanning tree) 
can be modified by exchanging some links, after the structure of the link costs has been changed. Typically, it is not possible to reconstruct the whole network, but only a given fraction of links can be exchanged in the current solution.  Models, which take such a situation into account,  belong the the class of  \emph{two-stage recoverable ones}. Usually, the parameters in the
 \emph{first stage} are precisely known, while the future, \emph{second-stage} values of the parameters are uncertain. However, the second-stage decision, called also a \emph{recourse action}, can be made after the true values of the parameters have been revealed.   

The idea of recoverable robustness, first proposed in~\cite{LLMS09}, is similar to the so-called 
\emph{adjustable robustness} in mathematical programming~\cite{BTG04}, in which the values of some decision variables  can be fixed after the true parameter realization has been revealed. It is also similar to the min-max-min approach, recently discussed in~\cite{BK17}, where a subset of solutions is computed in the first stage and one of them is adopted, when the second-stage scenario becomes known. In the recoverable approach a modification of the first-stage solution is allowed, while in the problems studied in~\cite{BK17} a complete solution can be chosen from a small subset of solutions chosen in the first stage.


In  the recoverable model a neighborhood of a given solution is defined. Each solution in this neighborhood is within some predefined distance from the current solution. For 0-1 programming problems, there are several natural ways of defining the distance function (see, e.g.,~\cite{SAO09}), and we will describe them in more detail in Section~\ref{secprob}.   A solution constructed in the first stage can be then replaced with a solution in its neighborhood in the second stage. The goal is to compute a pair of solutions, namely the \emph{first stage solution} and the \emph{second stage one} in its neighborhood, which minimize the total first and second stage costs. If the second stage costs are uncertain, then the robust approach can be used (see, e.g.,~\cite{BN09, KY97}) and we seek a pair which minimizes the total cost in a worst case. In this paper we will use the following \emph{polyhedral uncertainty representation}.  For each second stage cost we specify an interval of its possible values. An adversary can increase each cost within its interval but the total increase cannot be greater than a given \emph{budget}. Additional linear constraints, which model relationships between the costs, are also allowed. The resulting uncertainty set is a polyhedron, whose special case is the  interval budgeted uncertainty set, discussed in~\cite{CGKZ18,NO13}.  

The robust recoverable problems are a far-reaching generalization of the class of \emph{incremental problems} considered in~\cite{SAO09}. In an incremental problem we seek a cheapest solution within a neighborhood of a fixed solution. In~\cite{SAO09} the class of incremental network problem was examined. It has been shown that the incremental version of some basic network problems (for example, the shortest path) can be NP-hard for some natural neighborhood definitions. The robust recoverable network problems were first studied in~\cite{B11,B12}, where several negative complexity results were obtained. Recently, the robust recoverable versions of the minimum spanning tree~\cite{HKZ16a,HKZ16},  selection~\cite{CGKZ18,KZ15b}, and traveling salesperson~\cite{CG15b} problems were discussed. We will briefly describe complexity results for these problems in Section~\ref{seccompl}.

The goal of this paper is to provide a framework for solving robust recoverable 0-1 programming problems with a specified polyhedral uncertainty.
In general, such problems can be very complex from a computational point of view. The underlying deterministic single-stage problem can be already NP-hard and hard to approximate. Adding recoverable robustness leads to a min-max-min  0-1 programming problem, which can be very difficult to solve.  If the uncertainty set is finite, i.e. contains the 
finite numbers of scenarios, or can be replaced with its finite representation (for example by the set of extreme points of a polytope for polyhedral uncertainty), then a mixed integer programming (MIP for short) formulation can be built and solved by row and column generation techniques proposed, for example, in~\cite{ZZ13}. However, for the problems examined in this paper no finite representation of the considered uncertainty set is known. Hence, the solution method consisting in solving a MIP formulation can be hard to apply.  In Section~\ref{seclb} we propose several lower bounds, which will be based on solving one or a sequence of special MIP formulations. The formulations can be solved for quite large instances by using modern solvers. We will then use these lower bounds to characterize the quality of some approximate solutions in Section~\ref{secub}.  Finally, in Section~\ref{secexp}, we will present the results of the experiments for robust recoverable versions of the knapsack and assignment problems. The tested instances have large number of variables, so the proposed approach can be attractive in practical applications.

\section{Problem formulation}
\label{secprob}

Consider the following generic 0-1 programming problem~$\mathcal{P}$:
$$\begin{array}{llll}
		\min & \pmb{C}\pmb{x} \\
			& \pmb{x}\in \mathcal{X} \subseteq \{0,1\}^n,\\
	\end{array}
	$$
where $\pmb{C}=[C_1,\dots,C_n]$ is a vector of nonnegative costs and $\mathcal{X}$ is a set of feasible solutions. 
No additional assumptions about problem $\mathcal{P}$ are made here. In particular, $\mathcal{P}$ can be NP-hard and not at all approximable. Let $I(\pmb{x})=\{i\in [n]: x_i=1\}$, be the set of \emph{elements} contained in solution $\pmb{x}$ (we will 
denote by $[n]$ the set $\{1,\dots,n\}$).
Fix $\alpha\in [0,1]$, $\pmb{x}\in \mathcal{X}$, and let $\mathcal{X}^\alpha_{\pmb{x}}\subseteq \mathcal{X}$ be a \emph{neighborhood} of $\pmb{x}$. 
In this paper we will use the \emph{element exclusion neighborhood}~\cite{SAO09}, defined as follows:

\begin{equation}
\label{ndef}
\mathcal{X}^\alpha_{\pmb{x}}=\{\pmb{y}: |I(\pmb{x})\setminus I(\pmb{y})|\leq \alpha |I(\pmb{x})|\}=\{\pmb{y}\in \mathcal{X}: \sum_{i\in [n]} x_i(1-y_i)\leq \alpha\sum_{i\in [n]} x_i\}.
\end{equation}
Thus $\pmb{y}\in \mathcal{X}^\alpha_{\pmb{x}}$, if the number of elements which are in $I(\pmb{x})$ and not in $I(\pmb{y})$ is at most $\alpha |I(\pmb{x})|$.  In other words, forming solution $\pmb{y}$ from solution $\pmb{x}$, at most $\alpha|I(\pmb{x})|$ elements of $I(\pmb{x})$ can be removed.
%
Our definition of the neighborhoods slightly differs from the one considered in~\cite{SAO09}. However, in the context of the incremental problem, where $\pmb{x}$ (and thus $|I(\pmb{x})|$) is fixed, it is equivalent to that in~\cite{SAO09}. In the context of the recoverable approach, where $\pmb{x}$ is variable, it is more natural to provide a constant fraction $\alpha$, which is independent on $|I(\pmb{x})|$. Some other types of neighborhood, such as the 
\emph{element inclusion} or based on the \emph{symmetric difference}, were also examined in~\cite{SAO09}. The results presented in this paper can easily be modified for the aforementioned neighborhoods. It is sufficient to change the constraint,  describing the neighborhood in~(\ref{ndef}). Let $(\pmb{x}, \pmb{y})$, where $\pmb{x}\in \mathcal{X}$ and $\pmb{y}\in \mathcal{X}^\alpha_{\pmb{x}}$ be a \emph{feasible pair} of solutions. We will denote by $\mathcal{Z}$ the set of all such feasible pairs.

It is worth noting that all the neighborhoods,
based on the element inclusion,   the element exclusion and the symmetric difference, considered in~\cite{SAO09}
 are equivalent if $|I(\pmb{x})|=m$, $m\in [n]$, for each $\pmb{x}\in \mathcal{X}$. We will call such a problem $\mathcal{P}$,  an \emph{equal cardinality} problem.
We can then write
 $$\mathcal{X}^\alpha_{\pmb{x}}=\{\pmb{y}: |I(\pmb{x})\cap I(\pmb{y})|\geq  m(1-\alpha) \}=\{\pmb{y}\in \mathcal{X}: \sum_{i=1}^n x_iy_i\geq \ell\},$$
 where $\ell=\lceil m(1-\alpha) \rceil$ is a fixed integer, independent on solution $\pmb{x}$. 
  
Given $\pmb{x}\in \mathcal{X}$, and a vector of nonnegative costs $\pmb{c}=[c_1,\dots,c_n]$, we define the following \emph{incremental problem}:
$$\textsc{Inc}(\pmb{x},\pmb{c})=\min_{\pmb{y}\in \mathcal{X}^\alpha_{\pmb{x}}} \pmb{c}\pmb{y}.$$
In this problem we seek a solution $\pmb{y}$ in the neighborhood of $\pmb{x}$ of a minimum cost.  The incremental problem, for various network problems $\mathcal{P}$, was discussed in~\cite{SAO09}. In general, it can be harder than~$\mathcal{P}$ and its complexity will be presented in more detail in Section~\ref{seccompl}. Given a cost vector $\pmb{c}$, let us formulate the following \emph{recoverable problem}:
$$\textsc{Rec}(\pmb{c})= \min_{\pmb{x}\in \mathcal{X}} \min_{\pmb{y}\in \mathcal{X}^\alpha_{\pmb{x}}} (\pmb{C}\pmb{x}+\pmb{c}\pmb{y})=\min_{(\pmb{x},\pmb{y})\in \mathcal{Z}}(\pmb{C}\pmb{x}+\pmb{c}\pmb{y}).$$
In the recoverable problem we wish to find  a pair of solutions $\pmb{x}$ and $\pmb{y}\in  \mathcal{X}^\alpha_{\pmb{x}}$ (equivalently $(\pmb{x},\pmb{y})\in \mathcal{Z}$), which minimizes the total cost of both solutions. We will call $\pmb{x}$ the \emph{first stage solution} and $\pmb{y}$ the \emph{second stage solution}. The recoverable problem is a generalization of the incremental problem. Indeed, one can easily reduce the incremental problem for a given solution $\pmb{x}$ and cost vector $\pmb{c}$ to the recoverable problem by fixing $C_i=0$ if $x_i=1$ and $C_i=M$ if $x_i=0$, $i\in [n]$, where
 $M$ is a sufficiently large constant.

Suppose that the second stage cost vector $\pmb{c}$ is  uncertain, which means that precise cost values are not known 
in advance. Typically, uncertainty in parameters (costs) is modeled by specifying a set, denoted by 
$\mathcal{U}$, of all possible vectors of the parameter (cost) values, called \emph{scenarios}.
In this paper we focus on \emph{polyhedral uncertainty representation} and 
define the uncertainty set~$\mathcal{U}$ as follows.
 We are given a vector of \emph{nominal second stage costs} $\underline{\pmb{c}}=[\underline{c}_1,\dots,\underline{c}_n]$ and a vector $\pmb{d}=[d_1,\dots,d_n]$ of maximal \emph{deviations} of the costs from their nominal values. We also specify a \emph{budget} $\Gamma$, i.e. the amount of uncertainty, which can be allocated to the second stage costs. We consider the following polyhedral uncertainty set: 
\begin{equation}
\label{defunc}
\mathcal{U}=\{\underline{\pmb{c}}+\pmb{\delta}: \pmb{0}\leq \pmb{\delta}\leq \pmb{d}, ||\pmb{\delta}||_{1}\leq \Gamma, \pmb{\delta}\in\mathcal{V}\},
\end{equation}
where the vector $\pmb{\delta}=[\delta_1,\dots,\delta_n]$ represents deviations of the costs from their nominal values and $\mathcal{V}\subseteq \Rset^n$ is a polyhedron described by some additional linear constraints involving $\pmb{\delta}$. To ensure that $\underline{\pmb{c}}\in \mathcal{U}$, we will assume that $\pmb{0}\in \mathcal{V}$. One can optionally use $\mathcal{V}$ to model some additional relationships between the uncertainty of the costs. For example, a constraint $\sum_{i\in A} \delta_i\leq \Gamma_A$  can represent the situation in which a subset of the costs has its own budget $\Gamma_{A}\leq \Gamma$. Another constraints of the form $\alpha_{ij} \delta_i \leq \delta_j \leq \beta_{ij}\delta_i$, for some fixed $\alpha_{ij}\leq \beta_{ij}$, can model a possible correlation between $\delta_i$ and $\delta_j$. If $\mathcal{V}=\Rset^n$, then we get the following \emph{continuous budgeted interval uncertainty representation}~\cite{NO13}:
$$\mathcal{U}_0=\{(\underline{c}_i+\delta_i)_{i\in [n]}: \delta_i\in [0, d_i], \sum_{i\in[n]} \delta_i\leq \Gamma\}.$$
As $\mathcal{U}\subseteq \mathcal{U}_0$, the set $\mathcal{U}_0$ is a relaxation of $\mathcal{U}$. Given $\mathcal{U}$, we will investigate the following \emph{robust recoverable problem}:
$$\textsc{Rob-Rec}:\; \min_{\pmb{x}\in \mathcal{X}} \max_{\pmb{c}\in \mathcal{U}} \min_{\pmb{y}\in \mathcal{X}^\alpha_{\pmb{x}}} (\pmb{C}\pmb{x}+\pmb{c}\pmb{y}).$$
The robust recoverable problem generalizes the recoverable problem. It is enough to fix $\underline{\pmb{c}}=\pmb{c}$ and $\pmb{d}=\pmb{0}$ (or $\Gamma=0$). Assume that $\mathcal{X}_{\pmb{x}}^\alpha=\{\pmb{x}\}$. This is true, for example, when $\mathcal{P}$ is an equal cardinality problem and $\alpha=0$. If additionally $\pmb{C}=\pmb{0}$, then \textsc{Rob-Rec} can be rewritten as 
$$\min_{\pmb{y}\in \mathcal{X}} \max_{\pmb{c}\in \mathcal{U}} \pmb{c}\pmb{y},$$
which is a traditional \emph{single-stage robust min-max problem} widely discussed in the literature (see, e.g.,~\cite{KY97, BN09, BS03}). In this paper we will also consider the following \emph{evaluation}   problem:
$$\textsc{Eval}(\pmb{x})=\pmb{C}\pmb{x}+\max_{\pmb{c}\in \mathcal{U}} \min_{\pmb{y}\in \mathcal{X}^\alpha_{\pmb{x}}} \pmb{c}\pmb{y}=\pmb{C}\pmb{x}+\max_{\pmb{c}\in \mathcal{U}} \textsc{Inc}(\pmb{c},\pmb{x}),$$
where the inner maximization problem is called the \emph{adversarial problem}~\cite{NO13}.
Observe that solving \textsc{Rob-Rec} consists in minimizing $\textsc{Eval}(\pmb{x})$ over $\pmb{x}\in\mathcal{X}$. The evaluation problem generalizes the incremental one, as we get the latter problem by fixing $\pmb{C}=\pmb{0}$ and $\Gamma=0$.

\section{The computational complexity of problems}
\label{seccompl}

The computational  complexity of the incremental, recoverable, robust recoverable and adversarial problems is not less than the complexity of the generic problem $\mathcal{P}$. So, all these problems are NP-hard if $\mathcal{P}$ is already NP-hard. 
However, even the incremental version of $\mathcal{P}$ can be much harder than $\mathcal{P}$.
It has been shown in~\cite{SAO09}, then the incremental  shortest path problem ($\mathcal{X}$ is the set of characteristic vectors of simple paths in a given graph) for the element exclusion neighborhood is NP-hard and hard to approximate (interestingly, the incremental shortest path problem with the element inclusion neighborhood is polynomially solvable~\cite{SAO09}).
The incremental minimum assignment problem ($\mathcal{X}$ is the set of characteristic vectors of perfect matchings in a bipartite graph) is equivalent to the minimum exact matching problem, for which no polynomial time algorithm is known~\cite{SAO09}. 

The incremental problem is polynomially solvable, when $\mathcal{P}$ is the minimum spanning tree~\cite{SAO09}. It turns out that also the more general recoverable version of this problem, i.e. computing $\textsc{Rec}(\pmb{c})$ for a given scenario $\pmb{c}$, is polynomially solvable~\cite{HKZ16, HKZ16a}. In fact, the recoverable problem can be solved in polynomial time if $\mathcal{P}$ has a matroidal structure~\cite{HKZ16}. Namely, $\mathcal{X}$ is the set of characteristic vectors of the bases of some matroid.
 In particular, it can be solved efficiently when $\mathcal{P}$ is the selection problem~\cite{KZ15b}, i.e. when $\mathcal{X}=\{\pmb{x}\in \{0,1\}^n: x_1+\dots+x_n=p\}$ for a fixed constant $p\in [n]$.  Notice that that matroidal problems have the equal cardinality property, as each matroid base has the same number of elements. On the other hand, the recoverable version of the shortest path problem is NP-hard and hard to approximate for both element inclusion and element exclusion neighborhoods~\cite{B12,SAO09}.
 
 Consider now the \textsc{Rob-Rec} problem under the uncertainty set $\mathcal{U}$ discussed in this paper. Of course, this problem is not easier than $\textsc{Rec}(\pmb{c})$. So, it is interesting to characterize its complexity when the underlying recoverable problem is polynomially solvable. It has been recently shown that \textsc{Rob-Rec} is polynomially solvable under $\mathcal{U}_0$, when $\mathcal{P}$ is the selection problem~\cite{CGKZ18}. However, the complexity of \textsc{Rob-Rec} under $\mathcal{U}_0$ for other matroidal problems, in particular for the minimum spanning tree, remains open (the corresponding evaluation problem can be solved in polynomial time~\cite{NO13}). 
 
 Let us now briefly describe the known results for \textsc{Rob-Rec} under $\mathcal{U}_0$, with comparison to other uncertainty sets.  We will consider the \emph{discrete uncertainty set} $\mathcal{U}^D=\{\pmb{c}^1,\dots,\pmb{c}^r\}$, 
 where $\mathcal{U}^D$ contains  $r>1$ explicitly listed
cost scenarios $\pmb{c}^k$, $k\in [r]$,
 and 
 the \emph{discrete budgeted uncertainty set} $\mathcal{U}^1=\{ (\underline{c}_i+\delta_i)_{i\in [n]} : \delta_i \in [0,d_i], |\{i\in[n] : \delta_i>0\}| \leq \Gamma \}$ proposed in~\cite{BS03}. The known complexity results for basic problems $\mathcal{P}$ are summarized in Table~\ref{tab0}. 
\begin{table}[ht]
\centering
\footnotesize
\caption{Summary of complexity results for basic problem $\mathcal{P}$:  SEL - the selection problem, ST - 
the minimum spanning tree problem, SP (Incl.) - the shortest path  problem with the element inclusion neighborhood, SP (Excl.) - the shortest path problem  with the element exclusion neighborhood; \textbf{P} - polynomially solvable, \textbf{H} - NP-hard.} \label{tab0}
\begin{tabular}{l| ll | ll | ll | ll} 
& & &  \multicolumn{2}{c|}{$\mathcal{U}_0$} & \multicolumn{2}{c|}{$\mathcal{U}^1$} & \multicolumn{2}{c}{$\mathcal{U}^D$} \\
$\mathcal{P}$ 		 &   $\textsc{Inc}$ & $\textsc{Rec}$ & $\textsc{Eval}$ & $\textsc{Rob-Rec}$ & $\textsc{Eval}$ & $\textsc{Rob-Rec}$ & $\textsc{Eval}$ & $\textsc{Rob-Rec}$ \\ \hline
SEL & \textbf{P}~\cite{KZ15b} & \textbf{P}~\cite{KZ15b} & \textbf{P}~\cite{CGKZ18} & \textbf{P}~\cite{CGKZ18} & \textbf{P}~\cite{CGKZ18} & ? & \textbf{P}~\cite{KZ15b} & \textbf{H}~\cite{A01, KZ15b}  \\
ST & \textbf{P}~\cite{SAO09} & \textbf{P}~\cite{HKZ16, HKZ16a} & \textbf{P}~\cite{NO13} & ?  & \textbf{H}~\cite{LC93, NO13} & \textbf{H}~\cite{LC93, NO13} & \textbf{P}~\cite{SAO09} & \textbf{H}~\cite{KY97}\\
SP (Incl.) & \textbf{P}~\cite{SAO09} & \textbf{H}~\cite{B12} & \textbf{P}~\cite{NO13} & \textbf{H}~\cite{B12} & \textbf{H}~\cite{B12, BKS95, NO13} & \textbf{H}~\cite{B12, BKS95, NO13} & \textbf{P}~\cite{SAO09} & \textbf{H}~\cite{B12, KY97} \\
SP (Excl.) & \textbf{H}~\cite{SAO09} & \textbf{H}~\cite{SAO09} & \textbf{H}~\cite{SAO09} & \textbf{H}~\cite{SAO09} & \textbf{H}~\cite{SAO09}  & \textbf{H}~\cite{SAO09}& \textbf{H}~\cite{SAO09} & \textbf{H}~\cite{SAO09} 
\end{tabular}
\end{table}

Under $\mathcal{U}^D$, the evaluation problem has the same complexity as the incremental problem, because it reduces to solving $r$ incremental problem $\textsc{Inc}(\pmb{x},\pmb{c}^k)$ for $k\in [r]$. For this uncertainty representation, the 
single-stage robust min-max problem has been extensively discussed and all the negative results obtained in this area (see, e.g.,~\cite{A01, KY97}) remain valid for \textsc{Rob-Rec}. All the negative results presented in Table~\ref{tab0} remain true even if $r=2$. Observe that the complexity of  the \textsc{Rob-Rec} version of  the selection problem under $\mathcal{U}^1$ and the minimum spanning tree problem under $\mathcal{U}_0$ are interesting open problems. The complexity of the selection problem under more general set $\mathcal{U}$ of the form~(\ref{defunc}) is unknown. However, the problem is NP-hard when the uncertainty set is any polyhedron. It is enough to observe that the single-stage robust min-max problem for two scenarios $\pmb{c}^1$ and $\pmb{c}^2$ is equivalent to the single-stage robust
min-max problem under $\mathcal{U}'=\mathrm{conv}\{\pmb{c}^1,\pmb{c}^2\}$ and the former problem is known to be NP-hard~\cite{A01}. But, it is not clear that $\mathcal{U}'$ can be represented as~(\ref{defunc}).


\section{Solving the problems by MIP formulations}
\label{spmip}

The incremental and recoverable problems for the element exclusion neighborhood~(\ref{ndef}) can  be solved by using the following MIP formulations (\ref{eq0}a) and~(\ref{eq0}b), respectively:

\begin{equation}
\label{eq0}
\begin{array}{lll}
(a) \begin{array}{llll}
		\min & \pmb{c}\pmb{y} \\
			& \displaystyle \sum_{i\in [n]} x_i(1-y_i)\leq \alpha \sum_{i\in [n]} x_i\\
			&\pmb{y}\in \mathcal{X}\\
	\end{array}
	& (b)
	\begin{array}{llll}
			 \min &  \pmb{C}\pmb{x}+\pmb{c}\pmb{y} \\
			& \displaystyle \sum_{i\in [n]} (x_i-z_i)\leq \alpha \sum_{i\in [n]} x_i\\
			& z_i\leq x_i & i\in [n] \\
			& z_i\leq y_i & i\in [n]\\
			& z_i\geq x_i+y_i-1 & i\in [n]\\
			& z_i\in \{0,1\} & i\in [n]\\
			& \pmb{x}, \pmb{y}\in \mathcal{X}\\
	\end{array}
	\end{array}
\end{equation}

The MIP formulations can be simplified if $\mathcal{P}$ is an equal cardinality problem. The incremental and recoverable problems can be then formulated as follows:

\begin{equation}
\label{eq1}
\begin{array}{llll}
(a)
\begin{array}{llll}
		\min & \pmb{c}\pmb{y} \\
			& \sum_{i\in [n]} x_iy_i\geq  \ell\\
			& \pmb{y}\in \mathcal{X}\\
	\end{array}
	& (b)
	\begin{array}{llll}
		\min & \pmb{C}\pmb{x}+\pmb{c}\pmb{y} \\
			& \sum_{i\in [n]} z_i \geq \ell \\
			& z_i\leq x_i & i\in [n] \\
			& z_i\leq y_i & i\in [n]\\
			& z_i\geq 0 & i\in [n]\\
			& \pmb{x}, \pmb{y}\in \mathcal{X}\\
	\end{array}
	\end{array}
\end{equation}
Observe that we can drop the assumption that $z_i$ is binary in~(\ref{eq1}b).  The algorithms described in the next part of this paper will be based on the assumption that the formulations~(\ref{eq0}) and~(\ref{eq1}) can be solved exactly in reasonable time. To this purpose one can use a good off-the-shelf MIP solvers.

Consider now the evaluation problem, i.e. the problem of computing the value of $\textsc{Eval}(\pmb{x})$.
Given $\pmb{x}$, the inner adversarial problem $\max_{\pmb{c}\in \mathcal{U}} \textsc{Inc}(\pmb{x},\pmb{c})$ can be represented as the following linear programming problem.
\begin{equation}
\label{eq3a}
\begin{array}{llll}
		\max & t\\
			& t\leq \pmb{c}\pmb{y} & \forall \pmb{y} \in \mathcal{X}_{\pmb{x}}^\alpha\\
			& \pmb{c}\in \mathcal{U}
	\end{array}
\end{equation}
If $t_{opt}$ is the optimal value of $t$, then $\textsc{Eval}(\pmb{x})=\pmb{C}\pmb{x}+t_{opt}$.  Notice that~(\ref{eq3a}) is a linear programming problem, since $\pmb{c}\in \mathcal{U}$ can be described by a system of linear constraints with real variables and the set $ \mathcal{X}_{\pmb{x}}^\alpha$ is finite. However, formulation~(\ref{eq3a}) has exponential number of constraints. If we replace $\mathcal{X}_{\pmb{x}}^\alpha$ with a subset $\mathcal{Y} \subseteq \mathcal{X}_{\pmb{x}}^\alpha$ of feasible solutions, then we get an upper bound on $t_{opt}$. Also the value of $\textsc{Inc}(\pmb{x},\pmb{c})$, for any $\pmb{c}\in \mathcal{U}$, is a lower bound on $t_{opt}$. Thus
in order to find the value of $t_{opt}$ with a given accuracy $\epsilon\geq 0$, we can use a relaxation (constraint generation) algorithm shown in the form of Algorithm~\ref{alg1}.
\begin{algorithm}
\caption{Compute $\textsc{Eval}(\pmb{x})$ with accuracy $\epsilon$.} \label{alg1}
\begin{algorithmic}[1]
	\STATE $UB:=\infty$
	\STATE Choose an initial scenario $\pmb{c}_0\in \mathcal{U}$  \label{alg1_2}
	\STATE Solve $\textsc{Inc}(\pmb{x},\pmb{c}_0)$ obtaining $\pmb{y}^*\in \mathcal{X}_{\pmb{x}}^\alpha$, $LB:=\textsc{Inc}(\pmb{x},\pmb{c}_0)$
	\STATE $\mathcal{Y}:=\{\pmb{y}^*\}$
	\WHILE{$\frac{UB-LB}{LB}>\epsilon$}
		\STATE Solve the formulation~(\ref{eq3a}) with $\mathcal{X}_{\pmb{x}}^\alpha=\mathcal{Y}$ obtaining $(\pmb{c}^*,t^*)$, $UB:=t^*$ \label{alg1_6}
		\STATE Solve $\textsc{Inc}(\pmb{x},\pmb{c}^*)$ obtaining $\pmb{y}^*\in \mathcal{X}_{\pmb{x}}^\alpha$
		\STATE \textbf{if} $LB<\textsc{Inc}(\pmb{x},\pmb{c}^*)$ \textbf{then} $LB:=\textsc{Inc}(\pmb{x},\pmb{c}^*)$
		\STATE $\mathcal{Y}:=\mathcal{Y}\cup \{\pmb{y}^*\}$\label{alg1_9}
	\ENDWHILE
	\STATE \textbf{return} $\pmb{C}\pmb{x}+UB$.
\end{algorithmic}
\end{algorithm}

Algorithm~\ref{alg1} solves a sequence of the incremental problems and relaxed problems~(\ref{eq3a}). It is easily seen  that the algorithm converges. This fact follows from the observation that the size of $\mathcal{Y}$ increases by one at each step~\ref{alg1_9} of the algorithm. Indeed, suppose that $\pmb{y}^*$ is already present in $\mathcal{Y}$, so in formulation~(\ref{eq3a}) solved in
 step~\ref{alg1_6}. Then $LB=\textsc{Inc}(\pmb{x},\pmb{c}^*)=\pmb{c}^*\pmb{y}^*\geq \pmb{c}\pmb{y}^*$ for each $\pmb{c}\in \mathcal{U}$. In consequence $LB\geq t^*=UB$ and the algorithm terminates.

Notice that each relaxed problem~(\ref{eq3a}) is a  linear programming problem, which can be solved efficiently. Hence, the running time of the algorithm relies on the complexity of solving the incremental problem. For larger problems the algorithm may converge slowly. However, we can terminate it after a specified time is exceeded. In this case we get an upper bound on $\textsc{Eval}(\pmb{x})$. We can also improve the performance of the algorithm by choosing good initial cost scenario $\pmb{c}_0$ in 
step~\ref{alg1_2}. Such scenario will be proposed in Section~\ref{advlb}.

Finally, focus on the most complex \textsc{Rob-Rec} problem, which can be represented as the following program:
\begin{equation}
\label{miprec}
\begin{array}{llll}
	\min &  \pmb{C}\pmb{x}+\theta \\
		& \displaystyle \theta \geq \pmb{c}\pmb{y}_{\pmb{c}} & \pmb{c}\in \mathcal{U}\\
		& \displaystyle \sum_{i\in [n]} x_i (1-y_i^{\pmb{c}})\leq \alpha \sum_{i\in [n]} x_i & \pmb{c}\in \mathcal{U} \\
		& \pmb{x},\pmb{y}_{\pmb{c}}\in \mathcal{X} & \pmb{c}\in \mathcal{U}
\end{array}
\end{equation}

When $\mathcal{U}=\{\pmb{c}^1,\dots,\pmb{c}^r\}$ is explicitly given or  the uncertainty set can be replaced with its finite representation, for example by the extreme points of $\mathcal{U}$ (see, e.g.,~\cite{CG15b, ZZ13}), then~(\ref{miprec}) becomes a MIP formulation for \textsc{Rob-Rec}.  This formulation can have exponential number of variables and constraints and solving it requires special row and column generation techniques~\cite{ZZ13}. For the problem discussed in this paper the situation seems to be more complex, because it is difficult to replace $\mathcal{U}$ with its finite equivalent representation. In particular, we cannot replace $\mathcal{U}$ with the set of its extreme points. We will demonstrate this fact using the following simple example. Let $\mathcal{X}=\{(x_1,x_2)\in \{0,1\}^2: x_1+x_2=1\}$ and $\mathcal{U}=\{(0+\delta_1,0+\delta_2):  \delta_1,\delta_2\in [0,1], \delta_1+\delta_2\leq 1\}$. When $\pmb{C}=\pmb{0}$ and $\alpha=1$, we get the following problem
$$\max_{(c_1,c_2)\in \mathcal{U}} \min_{(x_1,x_2)\in \mathcal{X}} c_1x_1+c_2x_2.$$
This problem has a unique solution $c_1=0.5$, $c_2=0.5$ with the objective value equal to~0.5. The set of extreme points of $\mathcal{U}$ is $\{(0,0), (0,1), (1,0)\}$ and for each of these points the objective value is~0.
In this paper we do not consider MIP formulation for \textsc{Rob-Rec}. Instead, we will use the formulations~(\ref{eq0}), (\ref{eq1}), (\ref{eq3a}) for the incremental, recoverable and evaluation problems to construct approximate solutions for \textsc{Rob-Rec}. 


\section{Lower bounds}
\label{seclb}

In this section we will propose several methods of computing a lower bound for the \textsc{Rob-Rec} problem. We will then use these lower bounds to evaluate the quality of the approximate solutions. We will denote by $opt$ the optimal objective value in \textsc{Rob-Rec}.

\subsection{Adversarial lower bound}
\label{advlb}

It is 
easy to check
 that for each cost scenario $\pmb{c}\in \mathcal{U}$, the value of $\textsc{Rec}(\pmb{c})$ is a lower bound on $opt$. In consequence, the following adversarial problem can provide us the first general lower bound:
$$\textsc{Adv}:\; \max_{\pmb{c}\in\mathcal{U}} \textsc{Rec}(\pmb{c})=\max_{\pmb{c}\in \mathcal{U}} \min_{(\pmb{x},\pmb{y})\in \mathcal{Z}} (\pmb{C}\pmb{x}+\pmb{c}\pmb{y}).$$
Let us rewrite this problem as follows:
\begin{equation}
\label{eq4}
\begin{array}{llll}
		\max & t\\
			& t\leq \pmb{C}\pmb{x}+\pmb{c}\pmb{y} & \forall (\pmb{x},\pmb{y}) \in \mathcal{Z}\\
			& \pmb{c}\in \mathcal{U}
	\end{array}
\end{equation}
In order to solve~(\ref{eq4}) we will use a similar technique as for the model~(\ref{eq3a}). The corresponding algorithm is shown in the form of Algorithm~\ref{alg2}.
\begin{algorithm}
\caption{Compute $\textsc{Adv}$ with accuracy $\epsilon$.} \label{alg2}
\begin{algorithmic}[1]
	\STATE $UB:=\infty$
	\STATE Choose an initial scenario $\pmb{c}_0\in \mathcal{U}$
	\STATE Solve $\textsc{Rec}(\pmb{c}_0)$ obtaining $(\pmb{x}^*,\pmb{y}^*)\in \mathcal{Z}$ and $LB:=\textsc{Rec}(\pmb{c}_0)$
	\STATE $\mathcal{Z}':=\{(\pmb{x}^*,\pmb{y}^*)\}$
	\WHILE{$\frac{UB-LB}{LB}>\epsilon$}
		\STATE Solve the formulation~(\ref{eq4}) with $\mathcal{Z}:=\mathcal{Z}'$ obtaining $(\pmb{c}^*,t^*)$ and $UB:=t^*$.
		\STATE Solve $\textsc{Rec}(\pmb{c}^*)$ obtaining $(\pmb{x}^*,\pmb{y}^*)\in \mathcal{Z}$
		\STATE \textbf{if} $LB<\textsc{Rec}(\pmb{c}^*)$ \textbf{then} $LB:=\textsc{Rec}(\pmb{c}^*)$
		\STATE $\mathcal{Z}':=\mathcal{Z}'\cup \{(\pmb{x}^*,\pmb{y}^*)\}$.
	\ENDWHILE
	\STATE \textbf{return} $LB$.
\end{algorithmic}
\end{algorithm}

Again, (\ref{eq4}) is a linear programming problem, which can be solved efficiently for a small subset $\mathcal{Z}'\subseteq \mathcal{Z}$. In order to prove that  Algorithm~\ref{alg2} converges, one can use the same argument as in Section~\ref{spmip}. The running time of Algorithm~\ref{alg2} depends on the complexity of solving the recoverable problem, which is not easy in general. However, we can fix a limit for its running time, after which we still get a lower bound on $opt$.

For larger problems the relaxation algorithm may converge slowly. In order to speed up the computations we can start with a good heuristic initial scenario $\pmb{c}_0\in \mathcal{U}$, computed as follows:
\begin{equation}
\label{mheus}
	\begin{array}{lll}
		\max  & v \\
			& \underline{c}_i+\delta_i\geq \min \{\underline{c}_i+d_i,v\}\\
			& \sum_{i\in [n]} \delta_i\leq \Gamma\\
			& 0\leq \delta_i\leq d_i & i\in [n] \\
			& \pmb{\delta}\in \mathcal{V}
	\end{array}
\end{equation}
The idea of~(\ref{mheus}) is to uniformly distribute the budget among the smallest costs. This problem can be solved efficiently by using a binary search on $[0,V]$, where $V=\max_{i\in [n]} \{\underline{c}_i+d_i\}$.  An illustration for the uncertainty set $\mathcal{U}_0$ is shown in  Figure~\ref{figis}. In this case, given $v\geq 0$, we fix 
$\delta_i=\max\{0,\min\{d_i,v-\underline{c}_i\}\}$. We choose the maximum value of $v$ for which the constraint $\sum_{i\in [n]} \delta_i\leq \Gamma$ is satisfied. 
\begin{figure}[ht]
	\centering
	\includegraphics[height=6cm]{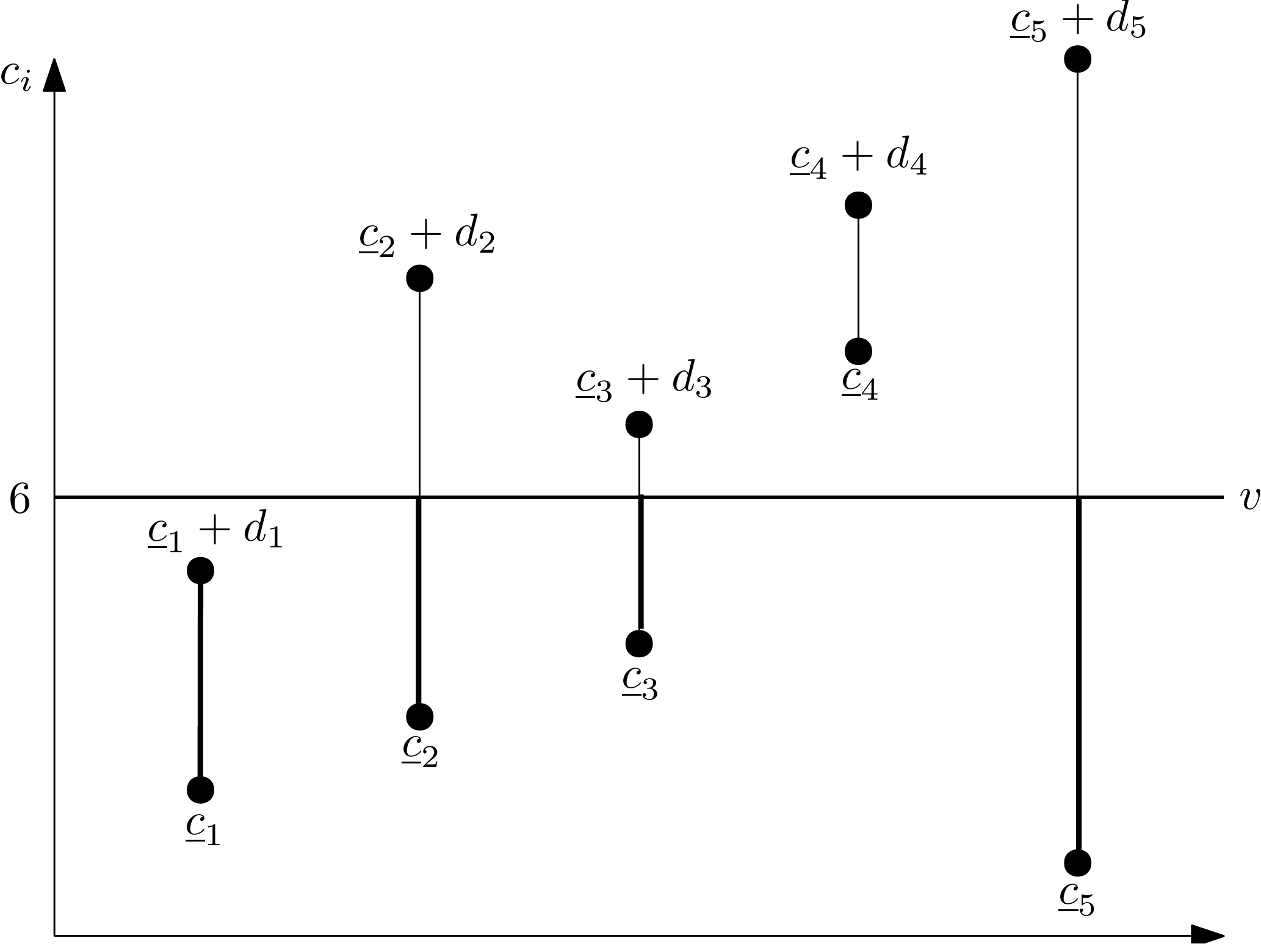}
	\caption{Computing the initial scenario $\pmb{c}_0$ for a sample uncertainty set $\mathcal{U}_0$ with $\Gamma=13$.}\label{figis}
\end{figure}

For large problems we can use $\textsc{Rec}(\pmb{c}_0)$ as a starting lower bound on $opt$. We can then try to improve the lower bound by running Algorithm~\ref{alg2} for a given time limit.

The tightness of the adversarial lower bound is likely to depend on $\alpha$. Observe that when $\alpha=1$, then $\mathcal{X}^\alpha_{\pmb{x}}=\mathcal{X}$ and  $\textsc{Rob-Rec}$ is equivalent to $\textsc{Adv}$. Also, for $\alpha$ close to~1, the adversarial lower bound should be closer to $opt$.
	
\subsection{Cardinality selection constraint  lower bound}
\label{sellb}

In this section we will propose another lower bound which, contrary to the adversarial lower bound, can be computed by solving one MIP formulation. This lower bound should behave better than the adversarial lower bound for smaller values of $\alpha$. In order to simplify the presentation we will use the uncertainty set $\mathcal{U}_0$. A generalization to any set $\mathcal{U}$ will be straightforward.
 The idea will be to relax the incremental problem by relaxing the structure of the neighborhood.  We consider first the case when $\mathcal{P}$ is an equal cardinality problem.
Let us replace the constraint $\pmb{y}\in\mathcal{X}$ in (\ref{eq1}a) with a weaker 
 \emph{cardinality (selection) constraint}, namely $y_1+\dots+y_n=m$, $m\in [n]$. So, the second stage solution need not to be feasible. Only, the cardinality  constraint must be satisfied. As the result, we get the following relaxation of the incremental problem:
\begin{equation}
\label{sel1}
\begin{array}{llll}
		\textsc{Inc}(\pmb{x},\pmb{c})\geq \textsc{Inc}'(\pmb{x},\pmb{c})=\min & \pmb{c}\pmb{y} \\
			&\displaystyle \sum_{i\in [n]} x_iy_i\geq  \ell\\
			& \displaystyle \sum_{i\in [n]} y_i=m \\
			& y_i\in \{0,1\} & i\in [n]
	\end{array}
\end{equation}
Since $\ell$ is integer and $\pmb{x}\in \{0,1\}^n$ is fixed, we get the following equivalent problem:
\begin{equation}
\label{sel2}
\begin{array}{llll}
		\textsc{Inc}'(\pmb{x},\pmb{c})=\min & \pmb{c}\pmb{y}\\
			&\displaystyle \sum_{i\in [n]} x_iy_i\geq \ell\\
			&\displaystyle \sum_{i\in [n]} y_i =m\\
			&0\leq y_i\leq 1  & i\in [n]
	\end{array}
\end{equation}
The problems~(\ref{sel1}) and~(\ref{sel2}) have the same optimal objective values, which follows from the fact that, given $\pmb{x}$, the constraint matrix of~(\ref{sel2}) is totally unimodular.
Taking the dual to~(\ref{sel2}) we get
$$
\begin{array}{llll}
		\textsc{Inc}'(\pmb{x},\pmb{c})=\max & \displaystyle \sigma \ell + \rho m-\sum_{i\in [n]} \alpha_i \\
			&\sigma x_i+\rho-\alpha_i\leq c_i & i\in [n]\\
			& \sigma \geq 0
	\end{array}
	$$
Now the relaxed evaluation problem $\textsc{Eval}'(\pmb{x})=\max_{\pmb{c}\in \mathcal{U}_0} \textsc{Inc}'(\pmb{x},\pmb{c})$ can be formulated as follows
$$
\begin{array}{llll}
		\textsc{Eval}(\pmb{x})\geq \textsc{Eval}'(\pmb{x})=\displaystyle \pmb{C}\pmb{x}+\max & \displaystyle \sigma \ell + \rho m-\sum_{i\in [n]} \alpha_i \\
			&\sigma x_i+\rho-\alpha_i \leq \underline{c}_i+\delta_i & i\in [n]\\
			& \displaystyle \sum_{i\in [n]} \delta_i\leq \Gamma \\
			& \delta_i\leq d_i & i\in [n] \\
			& \delta_i\geq 0 & i\in [n]\\
			& \sigma \geq 0
	\end{array}
	$$
Taking the dual to the inner maximization problem we get:
$$
\begin{array}{llll}
		\textsc{Eval}'(\pmb{x})=\pmb{C}\pmb{x}+\min & \displaystyle \pi\Gamma+\sum_{i\in [n]} \underline{c}_iy_i+\sum_{i\in [n]} u_i d_i \\
			& \displaystyle \sum_{i\in [n]} x_i y_i \geq \ell \\
			& \displaystyle \sum_{i\in [n]} y_i=m \\
			& y_i\leq 1 & i\in [n] \\
			& -y_i+\pi+u_i\geq 0 & i\in [n]\\
			& \pi\geq 0 \\
			& u_i, y_i\geq 0 & i\in[n]\\
	\end{array}
	$$
Finally, we obtain
\begin{equation}
\label{selfin0}
\begin{array}{llll}
		\displaystyle \min_{\pmb{x}\in \mathcal{X}} \textsc{Eval}'(\pmb{x})= \min & \displaystyle \pmb{C}\pmb{x}+\pi\Gamma+\sum_{i\in [n]} \underline{c}_iy_i+\sum_{i\in [n]} u_i d_i \\
			& \displaystyle \sum_{i\in [n]} x_i y_i \geq \ell \\
			& \displaystyle \sum_{i\in [n]} y_i=m \\
			& y_i\leq 1 & i\in [n] \\
			& -y_i+\pi+u_i\geq 0 & i\in [n]\\
			& \pi\geq 0 \\
			& \pmb{x}\in \mathcal{X}\\
			& u_i, y_i\geq 0 & i\in[n]\\
	\end{array}
\end{equation}
Since
 $$opt=\min_{\pmb{x}\in \mathcal{X}} \textsc{Eval}(\pmb{x})\geq \min_{\pmb{x}\in \mathcal{X}} \textsc{Eval}'(\pmb{x}),$$ the MIP formulation~(\ref{selfin0}) gives us a lower bound on $opt$. Notice that the constraint $\sum_{i\in [n]} x_iy_i\geq \ell$ can be easily linearized by using standard techniques.

Let us now turn to the element exclusion neighborhood. Let us remove the constraint $\pmb{y}\in\mathcal{X}$ in (\ref{eq0}a) and rewrite this problem as follows:
\begin{equation}
\label{esel0}
 \begin{array}{llll}
		\textsc{Inc}(\pmb{x},\pmb{c})\geq \textsc{Inc}'(\pmb{x},\pmb{c})=\min & \pmb{c}\pmb{y} \\
			& \displaystyle \sum_{i\in [n]} x_i(1-y_i)  \leq \alpha \sum_{i\in [n]} x_i \\
			& y_i\in \{0,1\} & i\in [n]
	\end{array}
\end{equation}
which can be rewritten as
\begin{equation}
\label{esel00}
 \begin{array}{llll}
		\textsc{Inc}'(\pmb{x},\pmb{c})=\min & \pmb{c}\pmb{y} \\
			& \displaystyle \sum_{i\in [n]} x_iy_i  \geq (1-\alpha) \sum_{i\in [n]} x_i \\
			& y_i\in \{0,1\} & i\in [n]
	\end{array}
\end{equation}
We cannot add the cardinality constraint, since the size of $I(\pmb{y})$ is unknown. Also, the right hand side of the constraint need not to be integral. However, we can still obtain a lower bound on the incremental problem by solving the following relaxation of~(\ref{esel0}):
$$
 \begin{array}{llll}
		 \textsc{Inc}(\pmb{x},\pmb{c})\geq\textsc{Inc}''(\pmb{x},\pmb{c})=\min & \pmb{c}\pmb{y} \\
			& \displaystyle \sum_{i\in [n]} x_iy_i  \geq (1-\alpha) \sum_{i\in [n]} x_i \\
			& y_i\leq 1 & i\in [n]\\
			& y_i\geq 0 & i\in [n]
	\end{array}
$$
Using similar reasoning as for the equal cardinality problem, we get the following MIP formulation:
\begin{equation}
\label{selfin1}
\begin{array}{llll}
		\displaystyle \min_{\pmb{x}\in \mathcal{X}} \textsc{Eval}''(\pmb{x})= \min & \displaystyle \pmb{C}\pmb{x}+\pi\Gamma+\sum_{i\in [n]} c_iy_i+\sum_{i\in [n]} u_i d_i \\
			& \displaystyle \sum_{i\in [n]} x_iy_i  \geq (1-\alpha) \sum_{i\in [n]} x_i\\
			& y_i\leq 1 & i\in [n] \\
			& -y_i+\pi+u_i\geq 0 & i\in [n]\\
			& \pi\geq 0 \\
			& \pmb{x}\in \mathcal{X}\\
			& u_i, y_i\geq 0 & i\in[n]\\
	\end{array}
\end{equation}
The formulation~(\ref{selfin1}) is a lower bound on $opt$ for the element exclusion neighborhood. The terms $x_iy_i$ in~(\ref{selfin1}) can be linearized by using standard techniques. 

Observe that for the equal cardinality problem $\mathcal{P}$ the cardinality selection constraint lower bound  is equal to~$opt$ when $\alpha=0$, because $\mathcal{X}^0_{\pmb{x}}=\{\pmb{x}\}$. The same property is true for the element exclusion neighborhood under the additional assumption that $\underline{\pmb{c}}>\pmb{0}$. In this case $\mathcal{X}^0_{\pmb{x}}$ contains all solution $\pmb{y}$ such that $I(\pmb{y})$ is a superset of $I(\pmb{x})$. Hence the cardinality selection constraint lower bound can behave better for $\alpha$ close to~0. 

\subsection{Lagrangian lower bound}
\label{laglb}

In this section we will construct another lower bound, which will be based on the Lagrangian relaxation technique (see, e.g.,~\cite{AMO93}). Contrary to the adversarial and selection lower bounds, this bound will be limited to a special class of problems. Namely, we will make two assumptions about the underlying problem $\mathcal{P}$. Firstly, we will assume that $\mathcal{X}=\{\pmb{x}\in \{0,1\}^n\,:\,\pmb{A}\pmb{x}=\pmb{b}\}$, where $\pmb{A}$ is an $m\times n$ matrix, and  the corresponding polyhedron $P_{\mathcal{X}}=\{\pmb{x}\,:\,\pmb{0}\leq \pmb{x} \leq \pmb{1}, \pmb{A}\pmb{x}=\pmb{b}\}$ is integral. This is true, for example, when $\pmb{A}$ is a totally unimodular matrix. We will also assume that $\mathcal{P}$ has the equal cardinality property, which will allow us to use the simplified neighborhood representation. An important  problem, which satisfies both assumptions,  is the minimum assignment. Again, we will consider the uncertainty set $\mathcal{U}_0$, as the generalization to any $\mathcal{U}$ is straightforward.

Let us introduce a Lagrangian multiplier $\mu\geq 0$ and consider the following Lagrangian relaxation of the incremental problem~(\ref{eq1}a):
\begin{equation}
\label{lag1}
\begin{array}{llll}
		\textsc{Inc}(\pmb{x},\pmb{c})\geq \textsc{Inc}'(\pmb{x},\pmb{c},\mu)=\min & \displaystyle \pmb{c}\pmb{y}+\mu(\ell- \sum_{i\in [n]} x_iy_i)\\
			& \pmb{A}\pmb{y}=\pmb{b}\\
			& \pmb{y}\in\{0,1\}^n
	\end{array}
\end{equation}
By the integrality property,~(\ref{lag1}) is equivalent to the following linear programming problem:
\begin{equation}
\label{lag2}
\begin{array}{llll}
		\textsc{Inc}'(\pmb{x}, \pmb{c},\mu)=\min & \displaystyle \sum_{i\in [n]} (c_i-\mu x_i)y_i+\mu\ell\\\
			& \pmb{A}\pmb{y}=\pmb{b}\\
			& y_i\leq 1 & i\in [n]\\
			& y_i\geq 0 & i\in [n]
	\end{array}
\end{equation}
Dualizing~(\ref{lag2}) yields
$$\begin{array}{llll}
		\textsc{Inc}'(\pmb{x},\pmb{c},\mu)=\max & \displaystyle \sum_{i\in [m]} \gamma_i b_i-\sum_{i\in [n]} \beta_i+\mu\ell\\\
			& \pmb{\gamma} \pmb{A}_i-\beta_i \leq c_i-\mu x_i & i\in [n] \\
			& \beta_i\geq 0 & i\in [n],
	\end{array}
	$$
where $\pmb{A}_i$ is the $i$th column of matrix $\pmb{A}$ and $\pmb{\gamma}=[\gamma_1,\dots,\gamma_m]$.
Now, given $\pmb{x}\in \mathcal{X}$, we get the following LP relaxation of the evaluation problem:
$$\begin{array}{llll}
		\textsc{Eval}(\pmb{x})\geq \textsc{Eval}'(\pmb{x},\mu)=\pmb{C}\pmb{x}+\max & \displaystyle \sum_{i\in [m]} \gamma_i b_i-\sum_{i\in [n]} \beta_i+\mu\ell\\\
			& \pmb{\gamma} \pmb{A}_i-\beta_i \leq \underline{c}_i+\delta_i-\mu x_i & i\in [n] \\
			& \displaystyle \sum_{i\in [n]} \delta_i\leq \Gamma \\
			& \delta_i\leq d_i & i\in [n]\\
			& \delta_i,\beta_i\geq 0 & i\in [n]
	\end{array}
	$$
After dualizing the inner maximization problem, we get the following equivalent formulation
$$\begin{array}{llll}
		\textsc{Eval}'(\pmb{x},\mu)=\pmb{C}\pmb{x}+\min & \displaystyle \pi\Gamma +\sum_{i\in [n]} y_i(\underline{c}_i-\mu x_i) +\sum_{i\in [n]} u_i d_i +\mu \ell \\
			& \pmb{A}\pmb{y}=\pmb{b} \\
			& -y_i+\pi+u_i \geq 0\\
			& 0\leq y_i \leq 1 & i\in [n]\\
			& u_i\geq 0 & i\in [n]
	\end{array}
	$$
Finally, we have
\begin{equation}
\label{lagfin}
\begin{array}{llll}
		\displaystyle \min_{\pmb{x}\in \mathcal{X}} \textsc{Eval}'(\pmb{x},\mu)=\min & \pmb{C}\pmb{x}+\displaystyle \pi\Gamma +\sum_{i\in [n]} y_i(\underline{c}_i-\mu x_i) +\sum_{i \in [n]} u_i d_i +\mu \ell \\
			& \pmb{A}\pmb{y}=\pmb{b} \\
			& -y_i+\pi+u_i \geq 0\\
			& y_i \leq 1 & i\in [n]\\
			& \pmb{A}\pmb{x}=\pmb{b} \\
			& u_i, y_i \geq 0 & i\in [n]\\
			& \pmb{x}\in \{0,1\}^n
	\end{array}
\end{equation}
The formulation~(\ref{lagfin}) can be linearized, which results in the following linear MIP model:
\begin{equation}
\label{lagfin1}
\begin{array}{llll}
		\displaystyle \min_{\pmb{x}\in \mathcal{X}} \textsc{Eval}'(\pmb{x},\mu)=\min & \pmb{C}\pmb{x}+\displaystyle \pi\Gamma +\sum_{i\in [n]} y_i\underline{c}_i- \mu \sum_{i\in [n]} z_i+\sum_{i\in [n]} u_i d_i +\mu \ell \\
			& \pmb{A}\pmb{y}=\pmb{b} \\
			& -y_i+\pi+u_i \geq 0\\
			& y_i \leq 1 & i\in [n]\\
			& \pmb{A}\pmb{x}=\pmb{b} \\
			& z_i\leq x_i & i\in [n] \\
			& z_i\leq y_i & i\in [n]\\
			& u_i, y_i,z_i \geq 0 & i\in [n]\\
			& \pmb{x}\in \{0,1\}^n
	\end{array}
\end{equation}
Observe that~(\ref{lagfin1}) has only $n$ binary variables.
Let us denote
$$\textsc{Eval}^*(\mu)= \min_{\pmb{x}\in \mathcal{X}} \textsc{Eval'}(\pmb{x},\mu)$$
Hence, for each $\mu\geq 0$, we get a lower bound on $opt$. The best lower bound can be computed by solving the following problem:
\begin{equation}
\label{lagprob}
\max_{\mu\geq 0} \textsc{Eval}^*(\mu).
\end{equation}
The problem~(\ref{lagprob}) can be solved by applying a search method on the single parameter $\mu \geq 0$. One can also solve one problem~(\ref{lagfin1}) for a heuristically chosen value of~$\mu$, also obtaining a lower bound on $opt$ (but possibly not the most tight one).

\section{Upper bounds and approximate solutions}
\label{secub}

As we  make no assumptions on the underlying problem $\mathcal{P}$, no general polynomial time approximation algorithm can exist for any problem discussed in this paper. In this section we will explore the approximability of the robust recoverable problem, under the assumption that we can solve the incremental and recoverable problems in reasonable time. In general, these problems cannot be solved in polynomial time. However, good modern solvers can solve them to optimality for quite large instances (see Section~\ref{secexp}). As in  Section~\ref{seclb}, we will use $opt$ to denote the optimal objective value for \textsc{Rob-Rec}.

By exchanging the inner max-min operators in \textsc{Rob-Rec}, we get the following \emph{min-max problem}:
$$\textsc{Min-Max}:\; \min_{(\pmb{x},\pmb{y})\in \mathcal{Z}} \max_{\pmb{c}\in \mathcal{U}} (\pmb{C}\pmb{x}+\pmb{c}\pmb{y})=\min_{(\pmb{x},\pmb{y})\in \mathcal{Z}} (\pmb{C}\pmb{x}+\max_{\pmb{c}\in \mathcal{U}} \pmb{c}\pmb{y}).$$
By using well known min-max relations, the optimal objective value of \textsc{Min-Max} is an upper bound on $opt$. Because $\mathcal{U}_0\subseteq \mathcal{U}$, we conclude that
$$UB=\min_{(\pmb{x},\pmb{y})\in \mathcal{Z}} (\pmb{C}\pmb{x}+\max_{\pmb{c}\in \mathcal{U}_0} \pmb{c}\pmb{y})$$
is also an upper bound on $opt$.
 Given $\pmb{y}$, the inner problem $\max_{\pmb{c}\in \mathcal{U}_0} \pmb{c}\pmb{y}$ can easily be solved by allocating the largest possible part of the budget $\Gamma$ to the costs of the elements in $I(\pmb{y})$. Hence, either the whole budget is allocated or the allocation is blocked by the upper bounds on the second stage costs of $I(\pmb{y})$. We thus get
$$\max_{\pmb{c}\in \mathcal{U}_0} \pmb{c}\pmb{y}=\min\{\underline{\pmb{c}}\pmb{y}+\Gamma, (\underline{\pmb{c}}+\pmb{d})\pmb{y} \}.$$
In consequence
$$
UB= \min_{(\pmb{x},\pmb{y})\in \mathcal{Z}} (\pmb{C}\pmb{x}+\min\{\underline{\pmb{c}}\pmb{y}+\Gamma, (\underline{\pmb{c}}+\pmb{d})\pmb{y} \})=\min_{(\pmb{x},\pmb{y})\in \mathcal{Z}}\min \{\pmb{C}\pmb{x}+\underline{\pmb{c}}\pmb{y}+\Gamma, \pmb{C}\pmb{x}+(\underline{\pmb{c}}+\pmb{d})\pmb{y} \}=$$
$$
	=\min \{\textsc{Rec}(\underline{\pmb{c}})+\Gamma, \textsc{Rec}(\underline{\pmb{c}}+\pmb{d})\}.
$$
Hence, in order to compute $UB$, it is enough to solve two recoverable problems.  We now investigate the quality of the solutions $(\underline{\pmb{x}},\underline{\pmb{y}})\in \mathcal{Z}$ and $(\overline{\pmb{x}},\overline{\pmb{y}})\in \mathcal{Z}$  obtained by solving $\textsc{Rec}(\underline{\pmb{c}})$ and $\textsc{Rec}(\underline{\pmb{c}}+\pmb{d})$, respectively. We will choose better of $\underline{\pmb{x}}$ and $\overline{\pmb{x}}$ as an approximate first stage solution to \textsc{Rob-Rec}, i.e. we choose $\underline{\pmb{x}}$ if $\textsc{Eval}(\underline{\pmb{x}})\leq \textsc{Eval}(\overline{\pmb{x}})$ and $\overline{\pmb{x}}$, otherwise. Observe that 
$$\textsc{Eval}(\underline{\pmb{x}})\leq \pmb{C}\underline{\pmb{x}}+\max_{\pmb{c}\in \mathcal{U}_0}\pmb{c}\underline{\pmb{y}} \leq \pmb{C}\underline{\pmb{x}}+\underline{\pmb{c}}\underline{\pmb{y}}+\Gamma=\textsc{Rec}(\underline{\pmb{c}})+\Gamma$$
and 
$$\textsc{Eval}(\overline{\pmb{x}})\leq \pmb{C}\overline{\pmb{x}}+\max_{\pmb{c}\in \mathcal{U}_0} \pmb{c}\overline{\pmb{y}}\leq
\pmb{C}\overline{\pmb{x}}+(\underline{\pmb{c}}+\pmb{d})\overline{\pmb{y}}=\textsc{Rec}(\underline{\pmb{c}}+\pmb{d}).$$ 
Hence
\begin{equation}
\label{appr1}
\min\{\textsc{Eval}(\underline{\pmb{x}}),\textsc{Eval}(\overline{\pmb{x}})\}\leq UB.
\end{equation}
Since $\textsc{Rec}(\pmb{c})$ is a lower bound on $opt$ for each $\pmb{c}\in \mathcal{U}$,  we  get 
\begin{equation}
\label{appr01}
\frac{UB}{opt}\leq \frac{\min \{\textsc{Rec}(\underline{\pmb{c}})+\Gamma, \textsc{Rec}(\underline{\pmb{c}}+\pmb{d})\}}{\textsc{Rec}(\pmb{c})}=\rho(\pmb{c}),\;\; \forall \pmb{c}\in \mathcal{U}.
\end{equation}
 Notice that we get the smallest ratio $\rho^*=\min_{\pmb{c}\in \mathcal{U}} \rho(\pmb{c})$ by choosing $\pmb{c}\in \mathcal{U}$ maximizing $\textsc{Rec}(\pmb{c})$, i.e. by solving the adversarial problem discussed in Section~\ref{advlb}. Using~(\ref{appr1}) we get
$$\min\{\textsc{Eval}(\underline{\pmb{x}}),\textsc{Eval}(\underline{\pmb{x}})\} \leq \rho^*\cdot opt,$$ 
so the better of solutions $\underline{\pmb{x}}$ and $\overline{\pmb{x}}$ has an approximation ratio of~$\rho^*$.

The value of $\rho^*$ depends on the problem data. Furthermore, its precise evaluation requires solving recoverable and adversarial problems, which can be time consuming.
We now show several estimations of the ratio $\rho^*$ from above, which can be computed more efficiently. We will use the following lemma:
\begin{lem}
Let $(\pmb{x}^*, \pmb{y}^*)\in \mathcal{Z}$ be an optimal solution to $\textsc{Rec}(\pmb{c}_0)$ for a fixed scenario $\pmb{c}_0\in \mathcal{U}$. We then have
\begin{equation}
\label{appr2}
\rho^*\leq  \min \left\{\frac{\pmb{C}\pmb{x}^*+\underline{\pmb{c}}\pmb{y}^*+\Gamma}{\pmb{C} \pmb{x}^*+ \pmb{c}_0\pmb{y}^*}, \frac{\pmb{C}\pmb{x}^*+(\underline{\pmb{c}}+\pmb{d})\pmb{y}^*}{\pmb{C} \pmb{x}^*+ \pmb{c}_0\pmb{y}^*}\right\}.
\end{equation}
\end{lem}
\begin{proof}
	Using~(\ref{appr01}), we get
	$$\rho^*\leq \rho(\pmb{c}_0)=\min\left\{ \frac{\textsc{Rec}(\underline{\pmb{c}})+\Gamma}{\text{Rec}(\pmb{c}_0)}, \frac{\textsc{Rec}(\underline{\pmb{c}}+\pmb{d})}{\textsc{Rec}(\pmb{c}_0)}\right\}$$
	Now~(\ref{appr2}) follows from the inequalities $\textsc{Rec}(\underline{\pmb{c}})\leq \pmb{C}\pmb{x}^*+\underline{\pmb{c}}\pmb{y}^*$, $\textsc{Rec}(\underline{\pmb{c}}+\pmb{d})\leq \pmb{C}\pmb{x}^*+(\underline{\pmb{c}}+\pmb{d})\pmb{y}^*$ and equality $\textsc{Rec}(\pmb{c}_0)=\pmb{C} \pmb{x}^*+ \pmb{c}_0\pmb{y}^*$.
\end{proof}
\begin{lem}
\label{lemappr3a}
	If $\frac{\underline{c}_i+d_i}{\underline{c}_i}\leq \sigma $, for each $i\in [n]$, then $\rho^*\leq \sigma$
\end{lem}
\begin{proof}
	By setting $\pmb{c}_0=\underline{\pmb{c}}\in \mathcal{U}$ in~(\ref{appr2}), we get $(\pmb{x}^*,\pmb{y}^*)=(\underline{\pmb{x}},\underline{\pmb{y}})$ and
	$$\rho^*\leq \frac{\pmb{C}\underline{\pmb{x}}+(\underline{\pmb{c}}+\pmb{d})\underline{\pmb{y}}}{\pmb{C}\underline{\pmb{x}}+\underline{\pmb{c}}\underline{\pmb{y}}}\leq \frac{\pmb{C}\underline{\pmb{x}}+(\underline{\pmb{c}}+\pmb{d})\underline{\pmb{y}}}{\pmb{C}\underline{\pmb{x}}+(1/\sigma)(\underline{\pmb{c}}+\pmb{d})\underline{\pmb{y}}}\leq \frac{\pmb{C}\underline{\pmb{x}}+(\underline{\pmb{c}}+\pmb{d})\underline{\pmb{y}}}{(1/\sigma)(\pmb{C}\underline{\pmb{x}}+(\underline{\pmb{c}}+\pmb{d})\underline{\pmb{y}})}=\sigma,$$
	where the second inequality follows from $(\underline{\pmb{c}}+\pmb{d})\underline{\pmb{y}}\leq \sigma \underline{\pmb{c}}\underline{\pmb{y}}$ and the third inequality follows form the fact  that $\frac{1}{\sigma}\leq 1$.
\end{proof}
The value of $\sigma$ in Lemma~\ref{lemappr3a} can be interpreted as the maximal factor by which the second stage costs can increase. For example, when $\sigma=2$, then the second stage costs can increase by at most 100\% from their nominal values, and in this case $\rho^*\leq 2$. It is reasonable to assume that in many practical applications $\sigma$ is not large, which results in good approximation ratio.

\begin{lem}
\label{lemappr1}
	If 
	$$\Gamma \leq \beta\cdot \left(\min_{\pmb{x}\in \mathcal{X}} \pmb{C}\pmb{x}+\max_{\pmb{c}\in \mathcal{U}} \min_{\pmb{y}\in \mathcal{X}} \pmb{c}\pmb{y}\right)$$
	for $\beta \geq 0$ then $\rho^*\leq 1+\beta$.
\end{lem}
\begin{proof}
	Let $\pmb{c}_0=\pmb{c}^*\in \mathcal{U}$ be scenario maximizing $\textsc{Rec}(\pmb{c})$ and
	let $(\pmb{x}^*,\pmb{y}^*)\in \mathcal{Z}$ be an optimal solution to $\textsc{Rec}(\pmb{c}^*)$.
	Using the first term in the minimum in~(\ref{appr2}) we get
	$$\rho^*\leq 1+\frac{\Gamma}{\pmb{C}\pmb{x}^*+\pmb{c}^*\pmb{y}^*}.$$
	Since
	$$\textsc{Rec}(\pmb{c}^*)=\pmb{C}\pmb{x}^*+\pmb{c}^*\pmb{y}^*= \max_{\pmb{c}\in \mathcal{U}}\left(\min_{\pmb{x}\in \mathcal{X}}\pmb{C}\pmb{x}+\min_{\pmb{y}\in \mathcal{X}_{\pmb{x}}^\alpha}\pmb{c}\pmb{y}\right)\geq \max_{\pmb{c}\in \mathcal{U}}\left(\min_{\pmb{x}\in \mathcal{X}}\pmb{C}\pmb{x}+\min_{\pmb{y}\in \mathcal{X}}\pmb{c}\pmb{y}\right)=$$
	$$
	=\min_{\pmb{x}\in \mathcal{X}} \pmb{C}\pmb{x}+\max_{\pmb{c}\in\mathcal{U}}\min_{\pmb{y}\in \mathcal{X}} \pmb{c}\pmb{y}$$
	the lemma follows.
\end{proof}
Lemma~\ref{lemappr1} shows that $\rho^*$ is not large if the budged is not large in comparison with the first and second stage solution costs. The value of $\min_{\pmb{x}\in \mathcal{X}} \pmb{C}\pmb{x}$ can be computed in polynomial time if $\mathcal{P}$ is polynomially solvable. Also, the value of $\max_{\pmb{c}\in\mathcal{U}}\min_{\pmb{y}\in \mathcal{X}} \pmb{c}\pmb{y}$ can be sometimes computed efficiently by dualizing the inner minimization problem and solving a resulting LP formulation. 

The next two lemmas are valid only for the uncertainty set $\mathcal{U}_0$.
\begin{lem}
\label{lemappr2}
	Assume that the uncertainty set is $\mathcal{U}_0$. If $$\Gamma \geq \beta\cdot \sum_{i\in [n]} d_i=\beta\cdot D$$
	for $\beta \in (0,1]$, then $\rho^*\leq \frac{1}{\beta}$.
\end{lem}	
\begin{proof}
	Choose scenario $\pmb{c}'\in \mathcal{U}_0$ such that $c_i'=\min\{\underline{c}_i+d_i, \underline{c}_i+\Gamma \frac{d_i}{D}\}$ for $i\in [n]$. It is clear that $\pmb{c}'\in \mathcal{U}_0$, because $\sum_{i=1}^n \delta'_i\leq \sum_{i=1}^n \Gamma \frac{d_i}{D}=\Gamma$. Let $(\pmb{x}^*,\pmb{y}^*)\in \mathcal{Z}$ be an optimal solution to $\textsc{Rec}(\pmb{c}')$.
	 Consider the second term in the minimum in~(\ref{appr2}). Because $\pmb{C}\pmb{x}^*\geq 0$ and $(\underline{\pmb{c}}+\pmb{d})\pmb{y}^*\geq \pmb{c}'\pmb{y}^*$, we get the following estimation:
	$$\rho^*\leq \frac{(\underline{\pmb{c}}+\pmb{d})\pmb{y}^*}{\pmb{c}'\pmb{y}^*}= \frac{\sum_{i\in [n]} (\underline{c}_i+d_i) y^*_i}{\sum_{i\in [n]} \min\{\underline{c}_i+d_i, \underline{c}_i+\Gamma\frac{d_i}{D}\}y^*_i}.$$
	Let $I_1=\{i\in I(\pmb{y}^*): \underline{c}_i+d_i\leq \underline{c}_i+\Gamma\frac{d_i}{D}\}$ and $I_2=I(\pmb{y}^*)\setminus I_1$. We get
	$$
	\rho^*\leq \frac{\sum_{i\in I_1\cup I_2} (\underline{c}_i+ d_i) }{\sum_{i\in I_1}(\underline{c}_i+d_i) +\sum_{i\in I_2} (\underline{c}_i+ \Gamma\frac{d_i}{D})}\leq \frac{\sum_{i\in I_1\cup I_2} (\underline{c}_i+ d_i) }{\sum_{i\in I_1}(\underline{c}_i+d_i) +\sum_{i\in I_2} (\underline{c}_i+ \beta d_i)}.$$
	Because $\beta\in (0,1]$,
	$$\rho^*\leq  \frac{\sum_{i\in I_1\cup I_2} (\underline{c}_i+ d_i)}{\beta(\sum_{i\in I_1}(\underline{c}_i+d_i) +\sum_{i\in I_2} (\underline{c}_i+d_i))}=\frac{1}{\beta}.$$
\end{proof}
Lemma~\ref{lemappr2} shows that $\rho^*$ is not large if the budget $\Gamma$ is not significantly smaller than $D$, which denotes the maximum amount of the uncertainty which can be allocated to the second stage item costs.

\begin{lem}
\label{lemappr3}
	Assume that the uncertainty set is $\mathcal{U}_0$. Let $q=|\{i\in [n]: d_i>0\}|$.
	Then, the inequality $\rho^*\leq q+1$ holds. Furthermore, if $\mathcal{P}$ is an equal cardinality problem and $d_i\geq \Gamma/n$ for each $i\in [n]$, then $\rho^*\leq \frac{n}{m}+1$.
\end{lem}
\begin{proof}
	Choose scenario $\pmb{c}'\in \mathcal{U}_0$ such that $c_i'=\min\{\underline{c}_i+d_i, \underline{c}_i+\frac{\Gamma}{q}\}$ for $i\in [n]$. 
	Indeed, $\pmb{c}'\in \mathcal{U}_0$, because $\sum_{i\in [n]} \delta'_i=\sum_{\{i\in [n]:d_i>0\}} \delta_i'\leq q\cdot \frac{\Gamma}{q}=\Gamma$.
	Let $(\pmb{x}^*, \pmb{y}^*)\in \mathcal{Z}$ be an optimal solution to $\textsc{Rec}(\pmb{c}')$.
	Using the first term in the minimum in~(\ref{appr2}) we get
	$$\rho^*\leq 1+\frac{\Gamma}{\pmb{C}\pmb{x}^*+\pmb{c}'\pmb{y}^*}\leq 1+\frac{\Gamma}{\pmb{c}'\pmb{y}^*}=1+\frac{\Gamma}{\sum_{i\in [n]} \min\{\underline{c}_i+d_i,\underline{c}_i+\frac{\Gamma}{q}\}y^*_i}.$$
	Let $I_1=\{i\in I(\pmb{y}^*): \underline{c}_i+d_i\leq \underline{c}_i+\frac{\Gamma}{q}\}$ and $I_2=I(\pmb{y}^*)\setminus I_1$. We get
	$$\rho^*\leq 1+\frac{\Gamma}{\sum_{i\in I_1} (\underline{c}_i+d_i) + \sum_{i\in I_2} (\underline{c}_i+\frac{\Gamma}{q})}.$$
	If $I_2=\emptyset$, then we obtain $\rho^*=1$ by using the second term in the minimum in~(\ref{appr2}). So $|I_2|\geq 1$ and we can estimate
	$$\rho^*\leq 1+\frac{\Gamma}{\frac{\Gamma}{q}}=1+q.$$
	If $d_i\geq \frac{\Gamma}{n}$ for each $i\in [n]$ and $\mathcal{P}$ is an equal cardinality problem, then  $q=n$, $I_1=\emptyset$, $|I_2|=m$ and
	$$\rho^*\leq 1+\frac{\Gamma}{\sum_{i\in I_2} \frac{\Gamma}{n}}=1+\frac{n}{m}.$$
\end{proof}
We can now apply Lemma~\ref{lemappr3} to several special cases of problem $\mathcal{P}$ under $\mathcal{U}_0$, with $d_i\geq \frac{\Gamma}{n}$ for each $i\in [n]$. If $\mathcal{P}$ is the selection problem discussed in~\cite{CGKZ18}, then $\rho^*\leq 1+\frac{n}{p}$. If $\mathcal{P}$ is the minimum spanning tree problem in a sparse graph, in which $|E|\leq \theta |V|$ for some constant $\theta\geq 1$, then $\rho^*\leq 1 +\frac{\theta |V|}{|V|-1}\sim 1+\theta$ for large graphs. If $\mathcal{P}$ is the minimum assignment problem, then $\rho^*\leq 1+\sqrt{n}$.

\section{Experiments}
\label{secexp}
In this section we will show the results of some experiments. 
We will test the lower bounds and approximate solutions using two problems, namely the assignment and the knapsack ones. The assignment problem is polynomially solvable and has the equal cardinality property. Hence we can apply all the lower bounds, proposed in Section~\ref{seclb}, to this problem. On the other hand, the knapsack problem is NP-hard and does not posses the equal cardinality property. In consequence, only the adversarial and cardinality selection constraint  lower bounds will be used for this problem. We will use scenario set $\mathcal{U}_0$, i.e. the continuous budgeted uncertainty.
The experiments were executed on
a \texttt{2 GHz} computer equipped with  80 \texttt{Intel(R) Xeon(R) CPU E7-4850} processors. 
We used \texttt{IBM ILOG CPLEX 12.8.0.0} optimizer~\cite{CPLEX} to solve the MIP formulations.

\subsection{The minimum assignment}

In this section we will show  the results of experiments when $\mathcal{P}$ is the following  assignment problem:
$$
	\begin{array}{llll}
		\min & \sum_{i\in [m]}\sum_{j\in [m]} C_{ij} x_{ij} \\
			& \sum_{i\in [m]} x_{ij}=1 & j\in [m] \\
			& \sum_{j\in [m]} x_{ij}=1 & i\in [m]\\
			& x_{ij}\in \{0,1\} & i,j\in [m]
	\end{array}
$$
The experiment was performed for $m\in \{10, 25,100\}$, so the number of variables $n\in\{100, 625,10~000\}$. The parameters were generated in the following way:
\begin{enumerate}
	\item The first stage costs $C_{ij}$, nominal second stage costs $\underline{c}_{ij}$ are random integers uniformly distributed in $[1,20]$.
	\item The maximal deviations $d_{ij}$ are random integers uniformly distributed  in $[0,100]$.
	\item The budget $\Gamma=0.1\sum_{i,j\in [n]} d_{ij}$, hence it is equal to 10\% of the total uncertainty of the second stage cots. 
		\item $\alpha\in \{0.1, 0.2, \dots,0.9\}$.
	\item The accuracy $\epsilon$ in Algorithm~\ref{alg1} and Algorithm~\ref{alg2} was set to 0.01 and both algorithms were terminated if the running time exceeds 600 seconds.
	 The accuracy of computing the Lagrangian lower bound by a version of golden search method was set to 0.1. The maximal time of solving the problem~(\ref{lagfin1}) was set to 600 seconds. After this time the computations of the bound were terminated.
\end{enumerate}

For each parameters setting, we have generated~10 random instances. 
In the first experiment we have computed, for every instance, the ratio 
\begin{equation}
\label{rhoh}
\rho(\pmb{c}_0)=\frac{\min\{\textsc{Rec}(\underline{\pmb{c}})+\Gamma, \textsc{Rec}(\underline{\pmb{c}}+\pmb{d})\}}{\textsc{Rec}(\pmb{c}_0)},
\end{equation}
 where $\pmb{c}_0$ is the heuristic scenario proposed in Section~\ref{advlb}. Computing this ratio requires solving three recoverable problems. Recall that the better of the solutions $\underline{\pmb{x}}$ or $\overline{\pmb{x}}$ has an approximation ratio at most $\rho(\pmb{c}_0)$.
 
 \begin{figure}[ht]
	\centering
	\includegraphics[height=5cm]{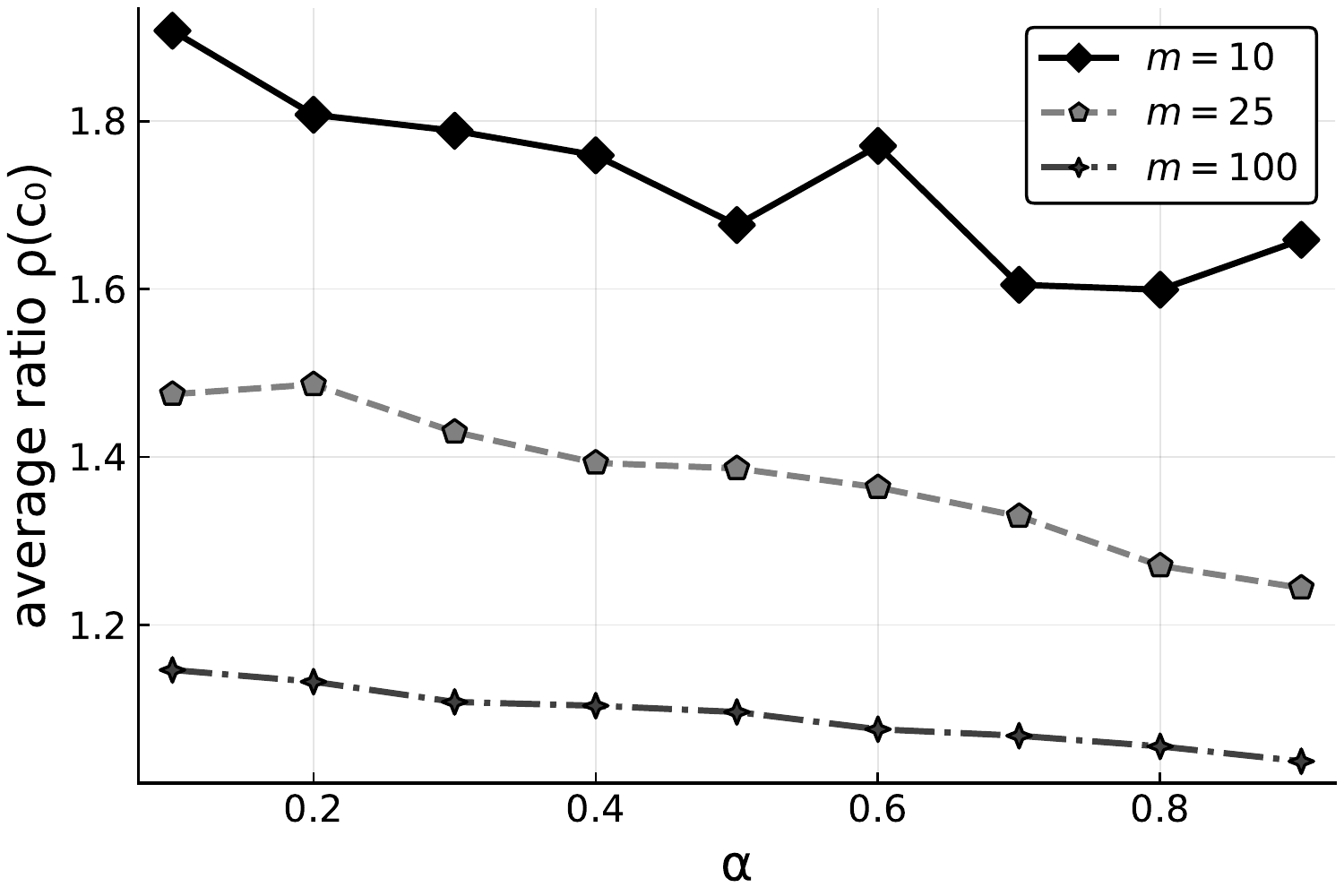}
	\includegraphics[height=5cm]{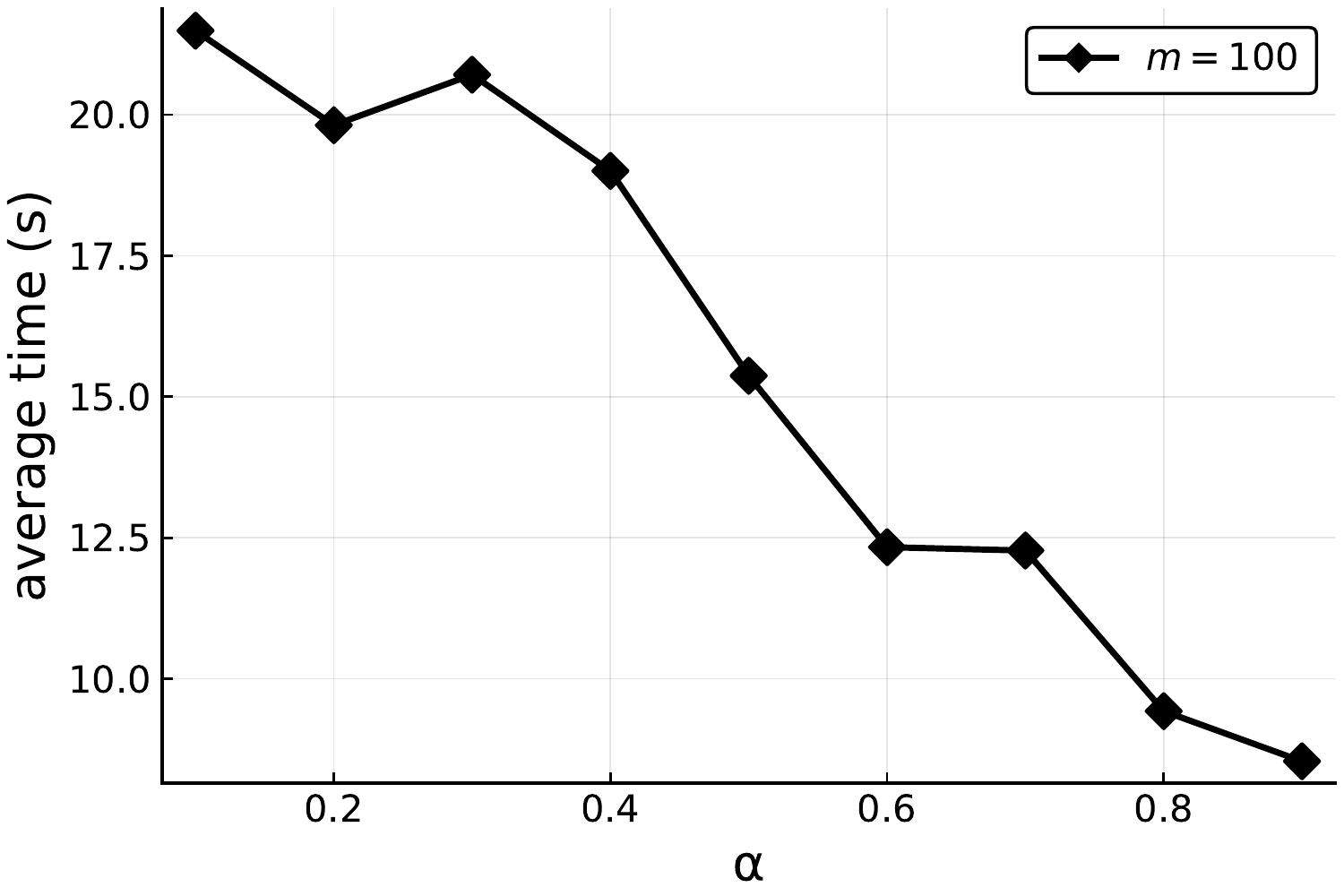}
	\caption{The average ratios $\rho(\pmb{c}_0)$ for the assignment problem with $m\in \{10,25,100\}$ and the average running times of computing $\rho(\pmb{c}_0)$ for $m=100$.} \label{fig4}
\end{figure}
 
 The average ratios $\rho(\pmb{c}_0)$ for various $m$ and the average times required to  compute it for $m=100$ are shown in Figure~\ref{fig4}. Observe first, that $\rho(\pmb{c}_0)$ can be computed efficiently. The average time required to compute $\rho(\pmb{c}_0)$, for $m=100$, is less than 25 seconds. It can be observed that the time is significantly smaller for larger~$\alpha$. The average ratios $\rho(\pmb{c}_0)$ are less than~2.  Interestingly, the ratio $\rho(\pmb{c}_0)$ is smaller for larger $m$ (for $m=100$ the average ratios $\rho(\pmb{c}_0)$ are less than~1.2). This fact is true for the particular method of data generation and may be different for other settings (verifying this requires more tests).

We now investigate the cases $m=10$ and $m=25$ in more detail. For each instance we computed: the adversarial lower bound $LB_{Adv}$ by executing Algorithm~\ref{alg2}, the lower bound $LB_{h}=\textsc{Rec}(\pmb{c}_0)$  for scenario $\pmb{c}_0\in \mathcal{U}$, proposed in Section~\ref{advlb}, the cardinality selection constraint  lower bound $LB_{Sel}$ by solving the MIP formulation~(\ref{selfin1}) constructed in Section~\ref{sellb} and the Lagrangian lower bound $LB_{Lag}$ constructed in Section~\ref{laglb}. Notice that $LB_{h}\geq LB_{Adv}$ for every instance, as we start Algorithm~\ref{alg2} from the initial scenario $\pmb{c}_0$.
We also computed the first stage solutions $\underline{\pmb{x}}$ and $\overline{\pmb{x}}$ (see Section~\ref{secub}), and the quantities $\textsc{Eval}(\underline{\pmb{x}})$ and $\textsc{Eval}(\overline{\pmb{x}})$ by using Algorithm~\ref{alg1}. We have computed the average ratios 
\begin{equation}
\label{rat0}
\rho_k=\frac{\min\{\textsc{Eval}(\underline{\pmb{x}}), \textsc{Eval}(\overline{\pmb{x}})\}}{LB_k},
\end{equation}
for each lower bound $LB_k$. We also measured the average running times of computing the ratios.

\begin{figure}[ht]
	\centering
	\includegraphics[height=5cm]{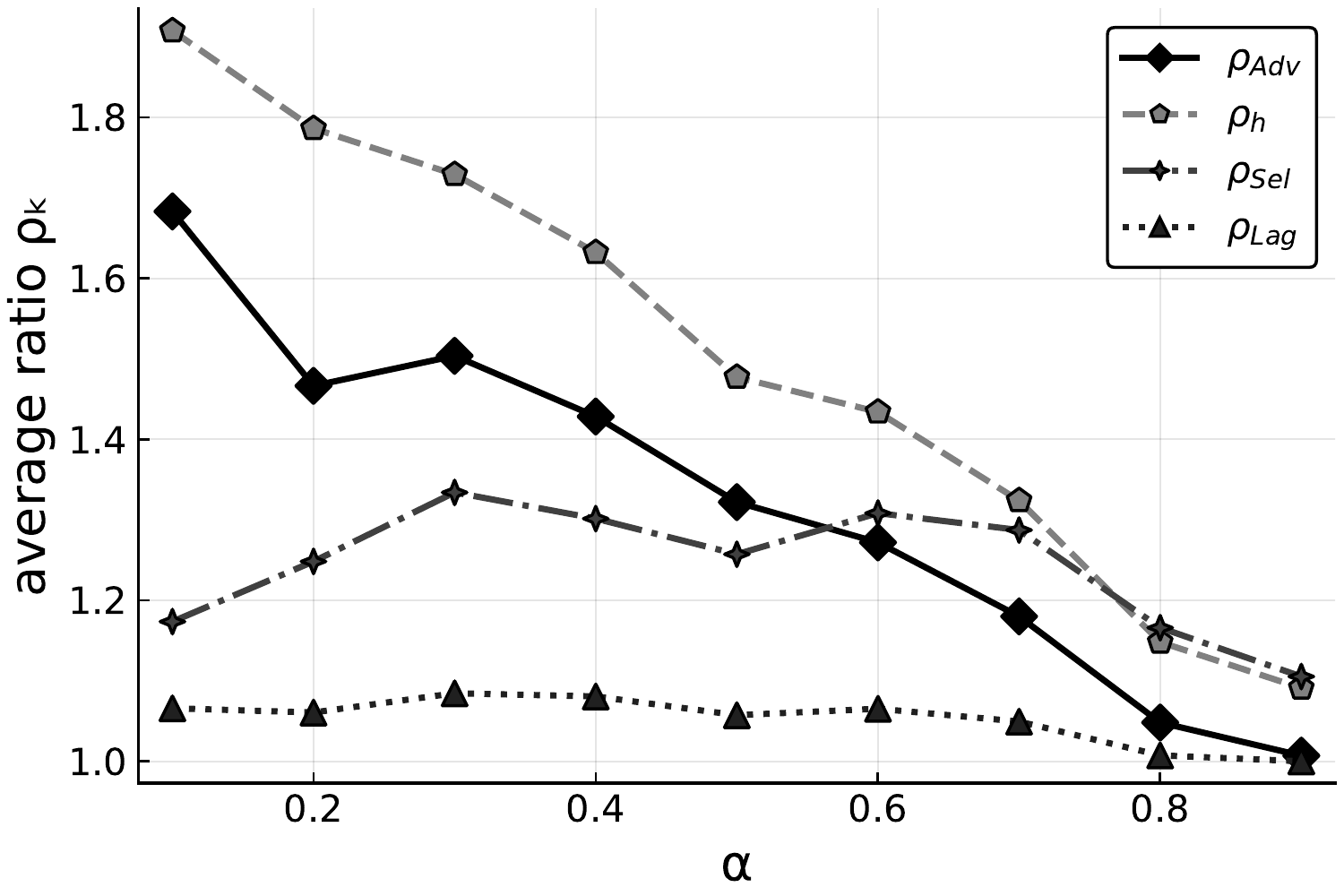}
	\includegraphics[height=5cm]{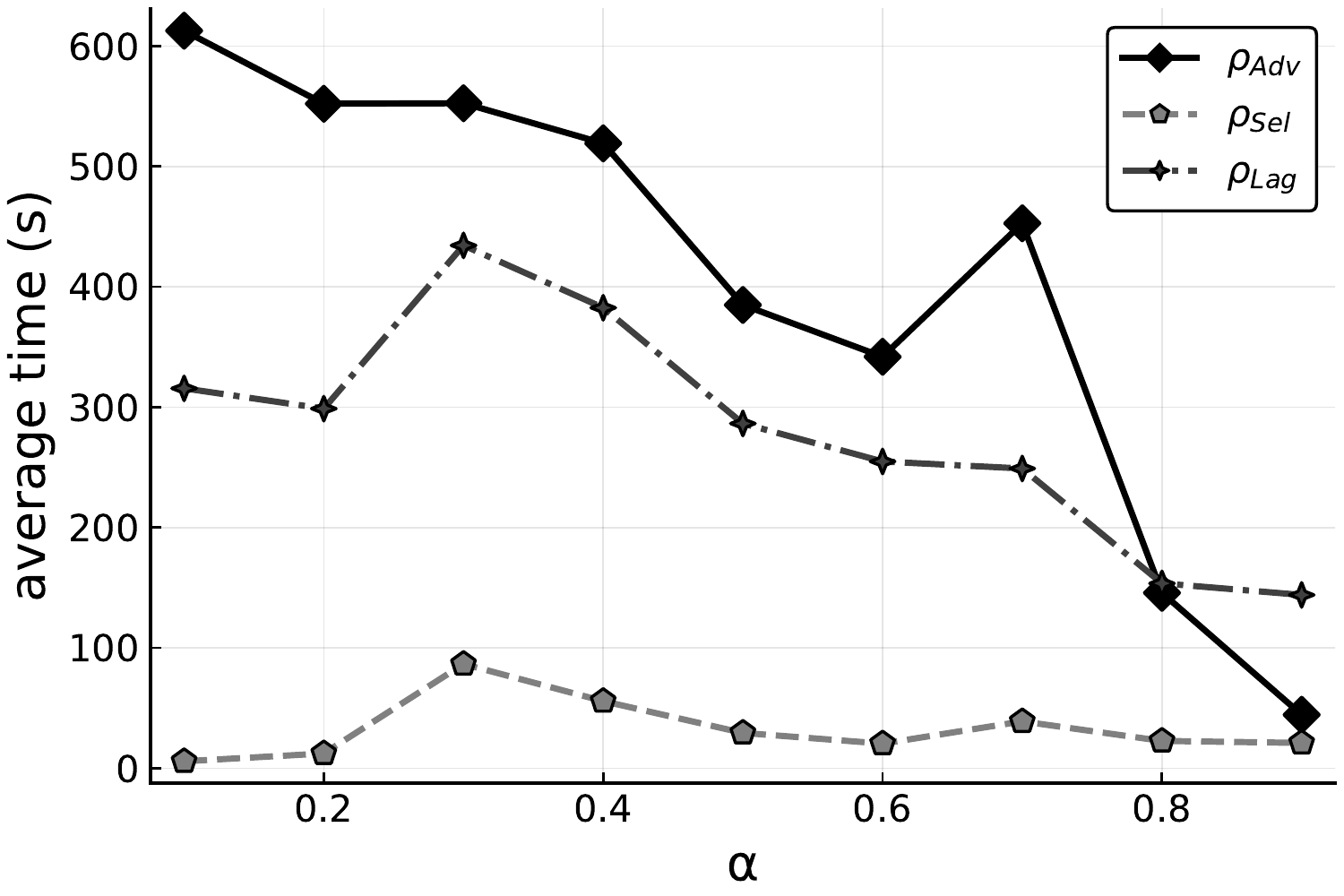}
	\caption{The average ratios $\rho_k$ and the average times for computing them for the assignment problem with $m=10$.} \label{fig2}
\end{figure}

The results for $m=10$ are shown in Figure~\ref{fig2}. For $m=10$, all the quantities were solved to optimality, or with the assumed accuracy $\epsilon$, when Algorithms~\ref{alg1} and~\ref{alg2} were used. One can observe in Figure~\ref{fig2} that  one of $\underline{\pmb{x}}$ or $\overline{\pmb{x}}$ is  always a good approximate solution. The best lower bound can be computed by using the Lagrangian relaxation technique (the bound $LB_{Lag}$). On can also see in Figure~\ref{fig2} that the evaluation and all lower bounds can be computed in reasonable time. The best lower bound is $LB_{Lag}$. As one can expect, the lower bound $LB_{Sel}$ is better than $LB_{Adv}$ for smaller $\alpha$ and worse for larger. Notice, however, that $LB_{Sel}$ can be computed very efficiently.

Figure~\ref{fig3} shows the results for $m=25$. 
We can still observe an improvement of $LB_{Adv}$ over $LB_{h}$. For this case, not all solutions $\overline{\pmb{x}}$ and $\underline{\pmb{x}}$ were evaluated exactly. For some instances Algorithm~\ref{alg2} was terminated after the time of 600 seconds was exceeded. In this case we obtained upper bounds on $\textsc{Eval}(\underline{\pmb{x}})$ and $\textsc{Eval}(\overline{\pmb{x}})$. The Lagrangian lower bound $LB_{Lag}$ was harder to compute than for $m=10$.  In Figure~\ref{fig3}, in the brackets the number of instances, for which  $LB_{Lag}$ was computed successfully, is shown. But, as for $m=10$, it outperforms all the remaining lower bounds and suggests that the approximate solutions behave well. The cardinality selection constraint lower bound $LB_{Sel}$ outperforms  $LB_{Adv}$  for $\alpha\leq 0.5$. The time required to compute $LB_{Sel}$ is again small.
\begin{figure}[ht]
	\centering
	\includegraphics[height=5cm]{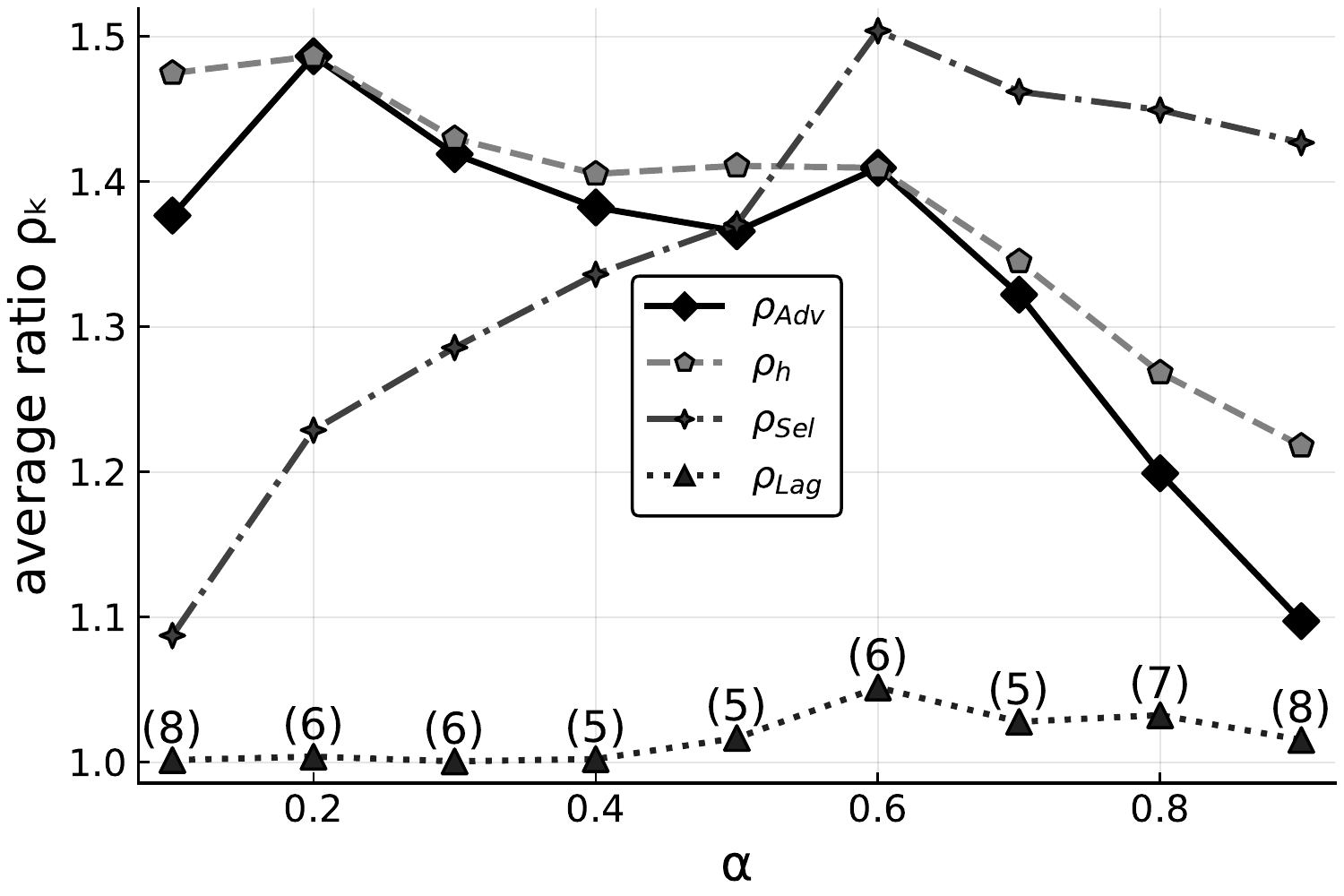}
	\includegraphics[height=5cm]{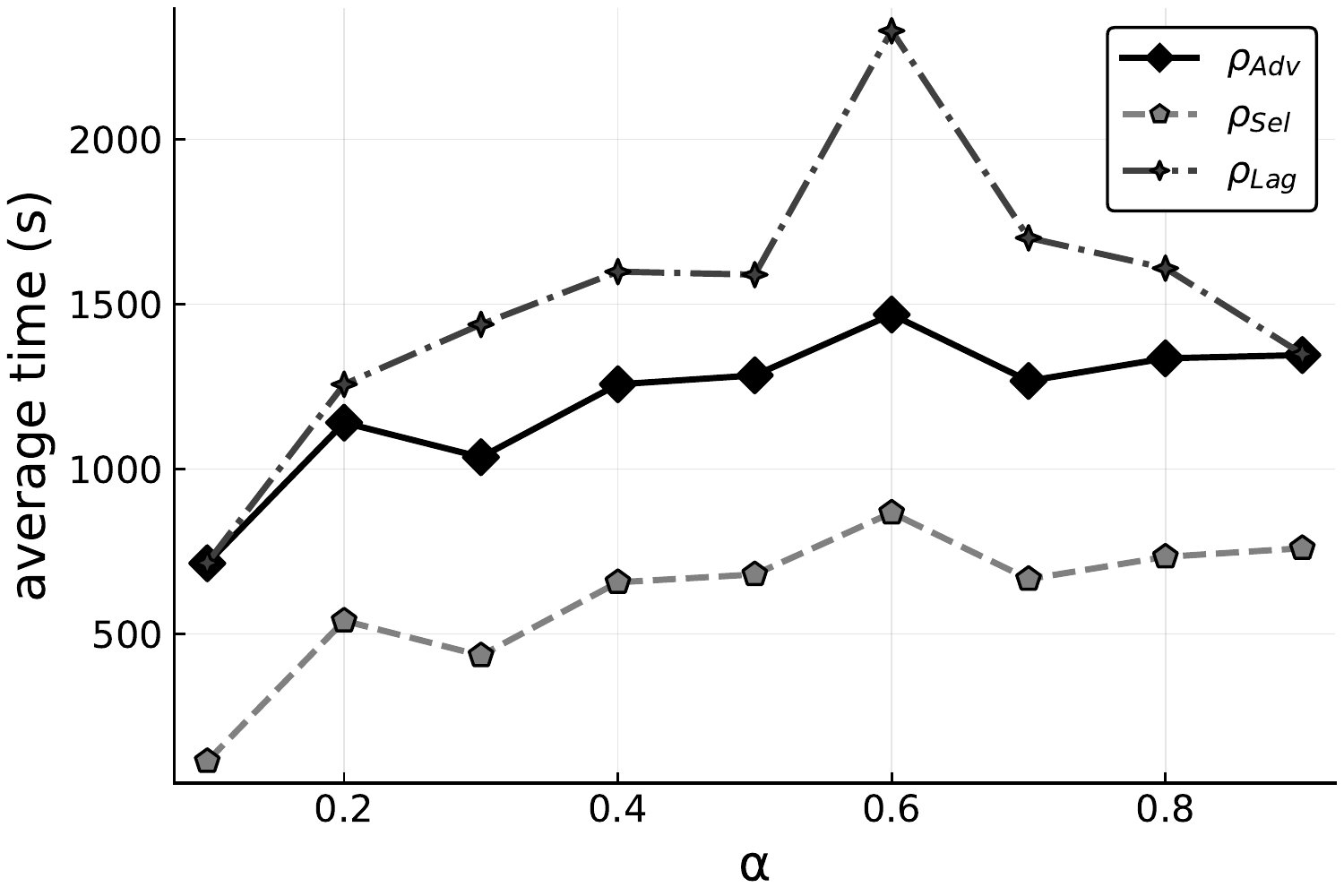}
	\caption{The average ratios $\rho_k$ and the average times for computing them for the assignment problem with $m=25$. The numbers in brackets denote the number of instances for which the value of $LB_{Lag}$ was computed successfully.} \label{fig3}
\end{figure}

\subsection{The minimum knapsack}

In this section we will show  the results of experiments when $\mathcal{P}$ is the following minimum knapsack problem:
$$
	\begin{array}{llll}
		\min & \sum_{i\in [n]} C_i x_i \\
			& \sum_{i\in [n]} w_i x_i\geq W \\
			& x_i\in \{0,1\} & i\in [n]
	\end{array}
$$
The test were performed for $n\in \{100, 400, 1000\}$, with the following parameter setting:
\begin{itemize}
	\item The first stage costs $C_i$, nominal second stage costs $\underline{c}_i$, and weights $w_i$ are random integers uniformly distributed in $[1,20]$. The knapsack capacity $W=0.3\sum_{i\in [n]} w_i$.
	\item The maximal deviations $d_i$ are random integers uniformly distributed in $[0,100]$.
	\item The budget $\Gamma=0.1\sum_{i\in [n]} d_i$.
	\item $\alpha\in \{0.1, 0.2, \dots,0.9\}$.
	\item The accuracy in Algorithm~\ref{alg1} and Algorithm~\ref{alg2} was set to 0.01. Algorithm~\ref{alg1} and Algorithm~\ref{alg2}  were terminated after the time limit of 600 seconds was exceeded. Also, the time limit on the MIP formulation~(\ref{selfin1}), for computing the cardinality selection constraint  lower bound, was set to 600 seconds. If this time was exceeded, then an estimation from below for this lower bound was returned.
\end{itemize}

For each parameters settings, we have generated~10 random instances. In the first experiment we computed the ratio $\rho(\pmb{c}_0)$ by using~(\ref{rhoh}). The average ratios and the average running time of computing them for $n=1000$ are shown in Figure~\ref{fig7}.

\begin{figure}[ht]
	\centering
	\includegraphics[height=5cm]{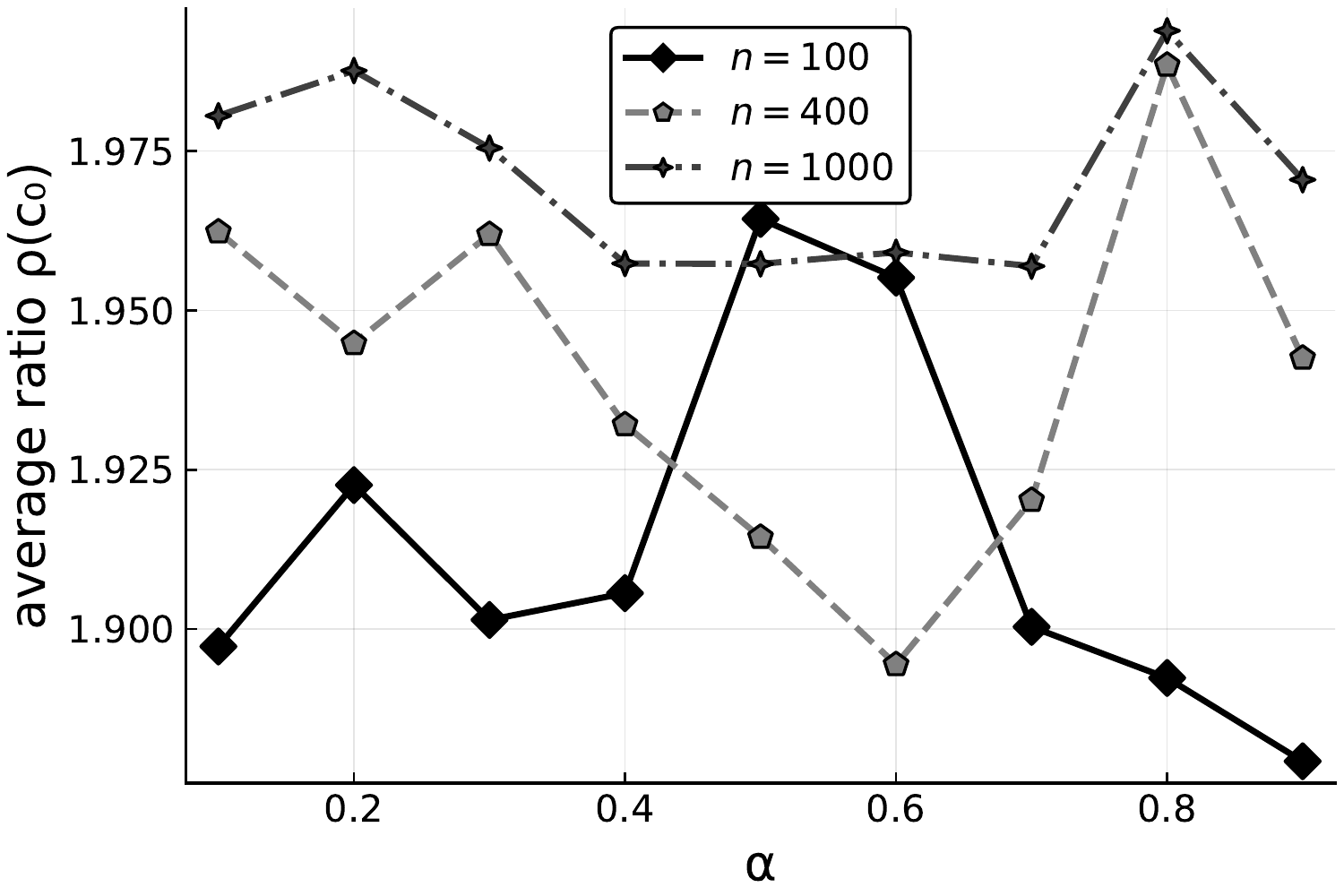}
	\includegraphics[height=5cm]{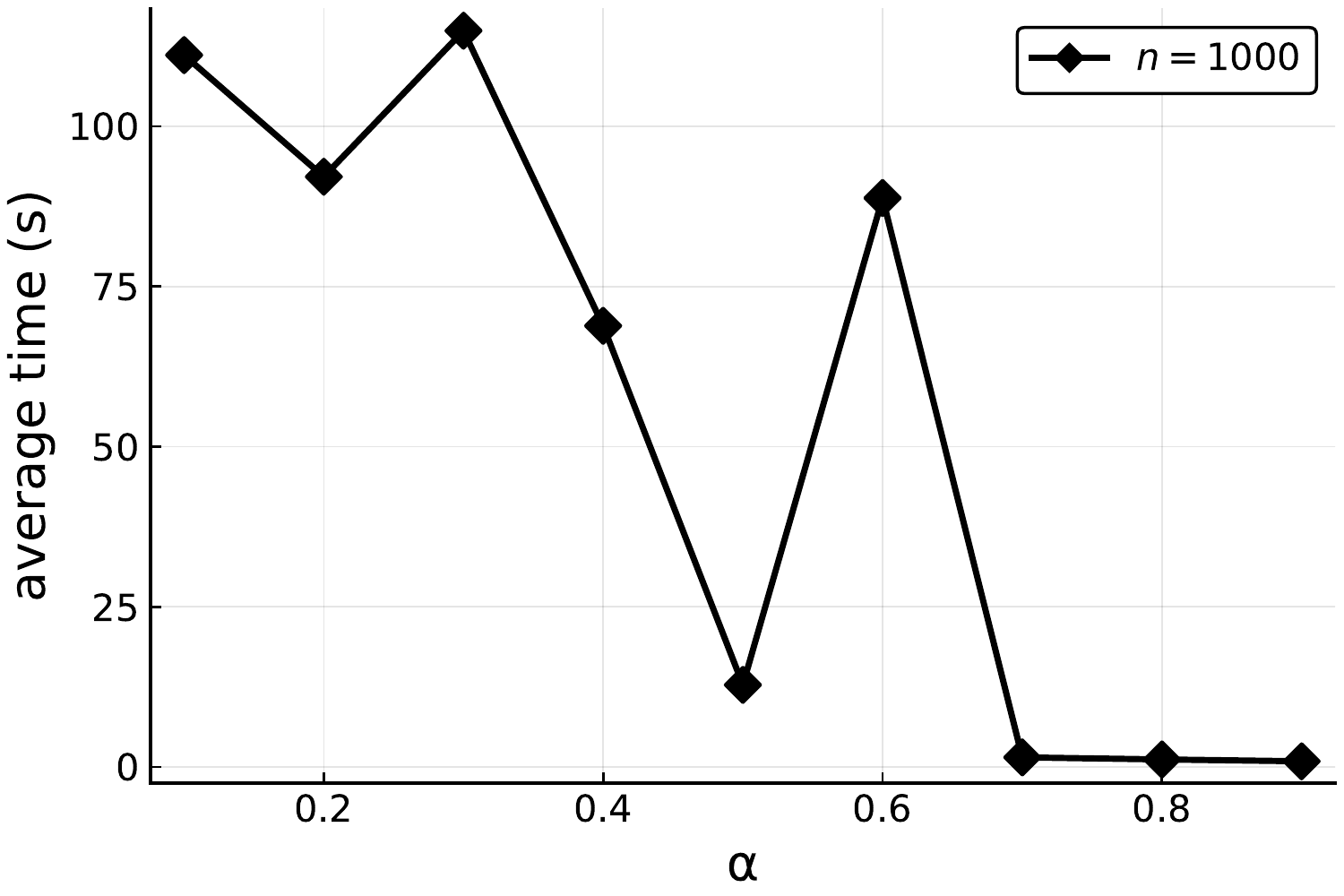}
	\caption{The average ratios $\rho(\pmb{c}_0)$ for the knapsack problem with $n\in\{100,400,1000\}$ and the average times of computing $\rho(\pmb{c}_0)$ for the knapsack problem with $n=1000$.} \label{fig7}
\end{figure}

Observe first that the ratio $\rho(\pmb{c}_0)$ can be computed efficiently. The largest running times were observed for $\alpha=0.3$. The average value of this ratio is less than 2.0 and for smaller $n$ the figure is more chaotic. For $n=1000$, the average value of $\rho(\pmb{c}_0)$ is close to 1.975 for all $\alpha$. This behavior is different than for the assignment problem (see Figure~\ref{fig4}), where the ratio is significantly smaller for larger instances and slightly decreases when $\alpha$ increases. 

\begin{figure}[ht]
	\centering
	\includegraphics[height=5cm]{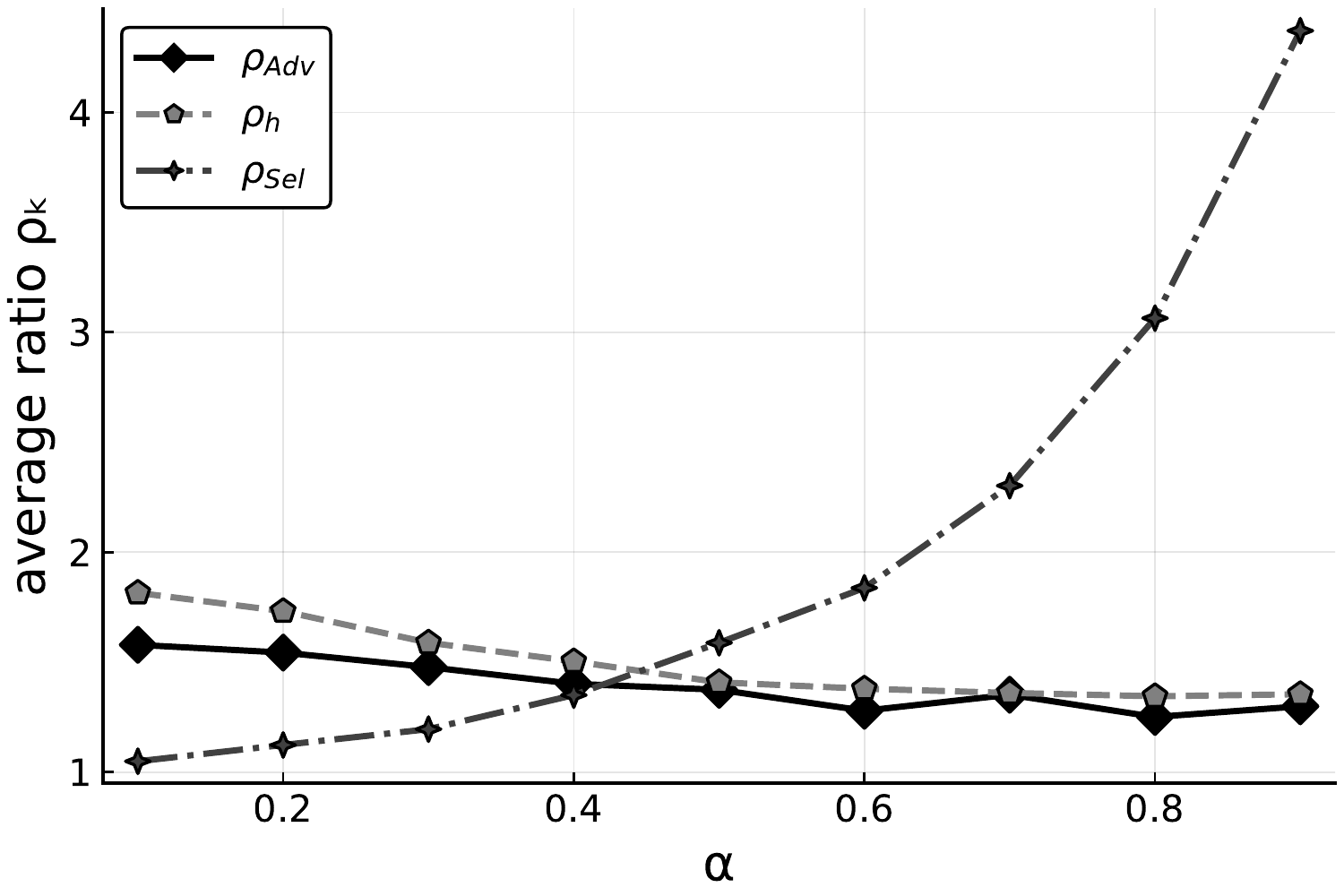}
	\includegraphics[height=5cm]{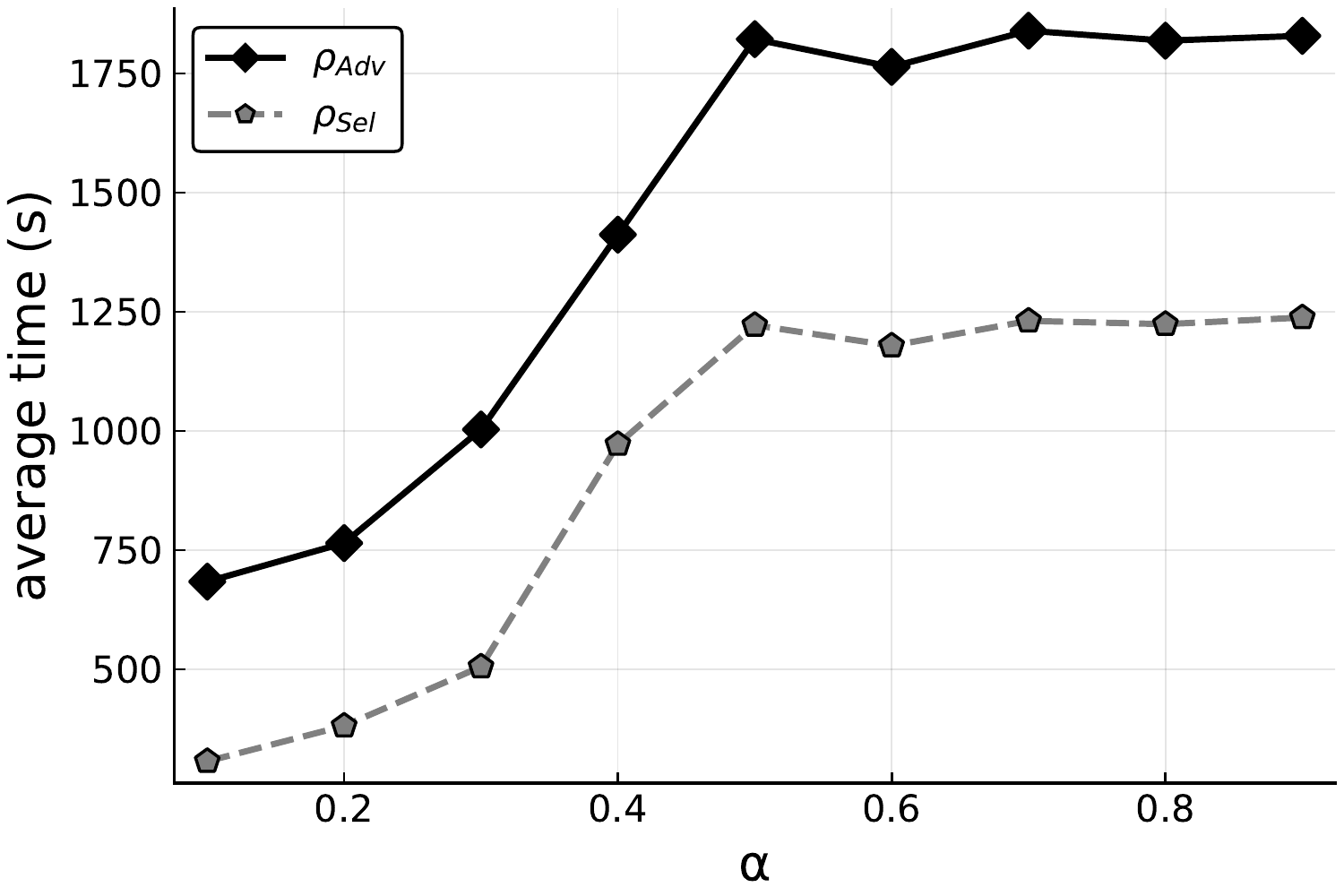}
	\caption{The average ratios $\rho_k$ and the average times for computing them for the knapsack problem with $n=100$.} \label{fig9}
\end{figure}

We next considered the case with $n=100$.
Figure~\ref{fig9} shows the average ratios $\rho_{Adv}$, $\rho_h$ and $\rho_{Sel}$ for $n=100$ (see~(\ref{rat0})) and the average times for computing these ratios for various $\alpha$. One can observe that the approximation algorithm proposed in Section~\ref{secub} performs well for the tested instances. By using better of $LB_{Adv}$ and $LB_{Sel}$, the average ratio for each $\alpha$ was not greater than 1.5. There is also an improvement of $\rho_{Adv}$ over $\rho_{h}$. The cardinality selection constraint  lower bound is better than the adversarial one for $\alpha<0.4$ and worse for $\alpha>0.4$. Observe that computing $\rho_{Sel}$ for the knapsack problem is more time consuming than for the assignment.

\begin{figure}[ht]
	\centering
	\includegraphics[height=5cm]{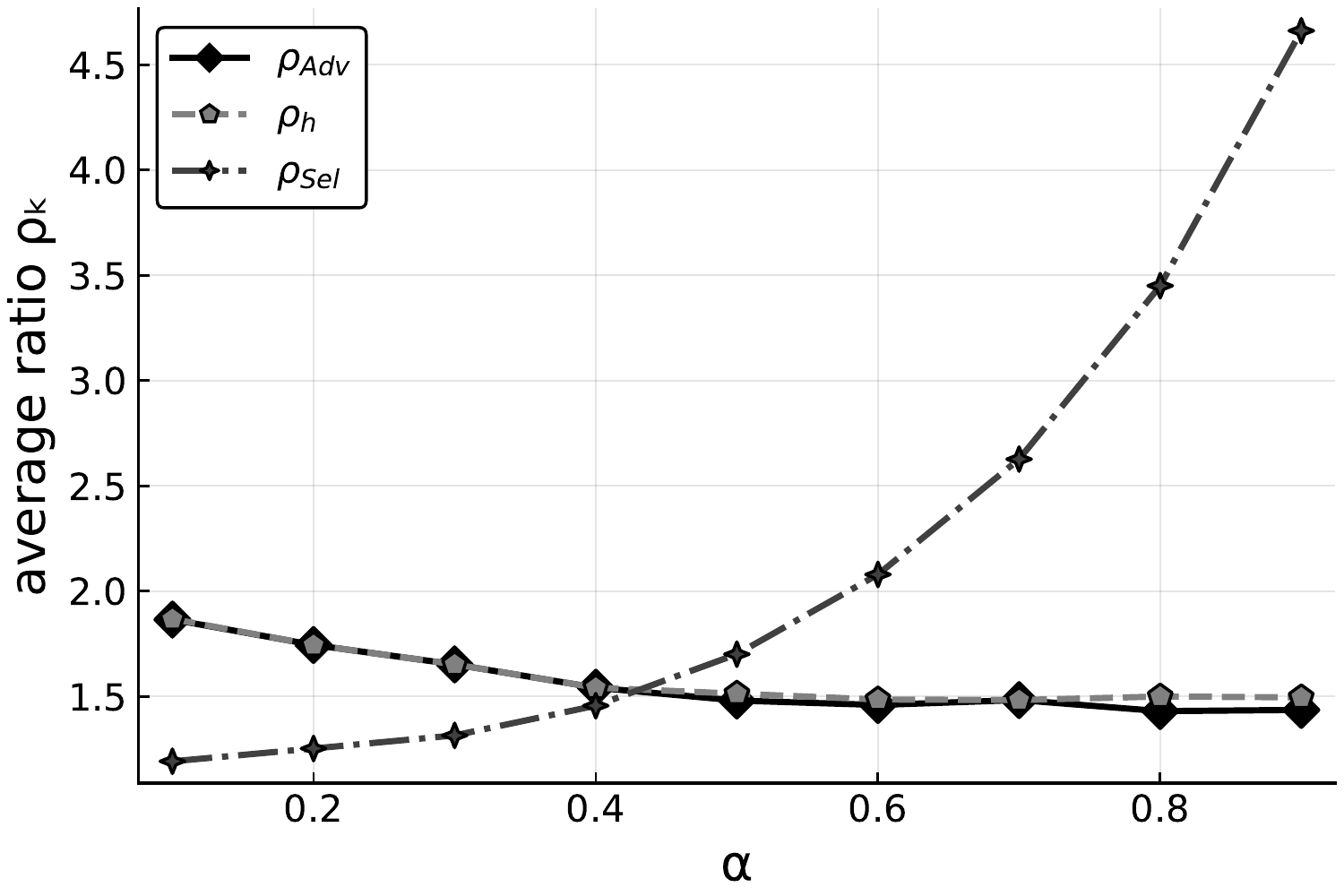}
	\includegraphics[height=5cm]{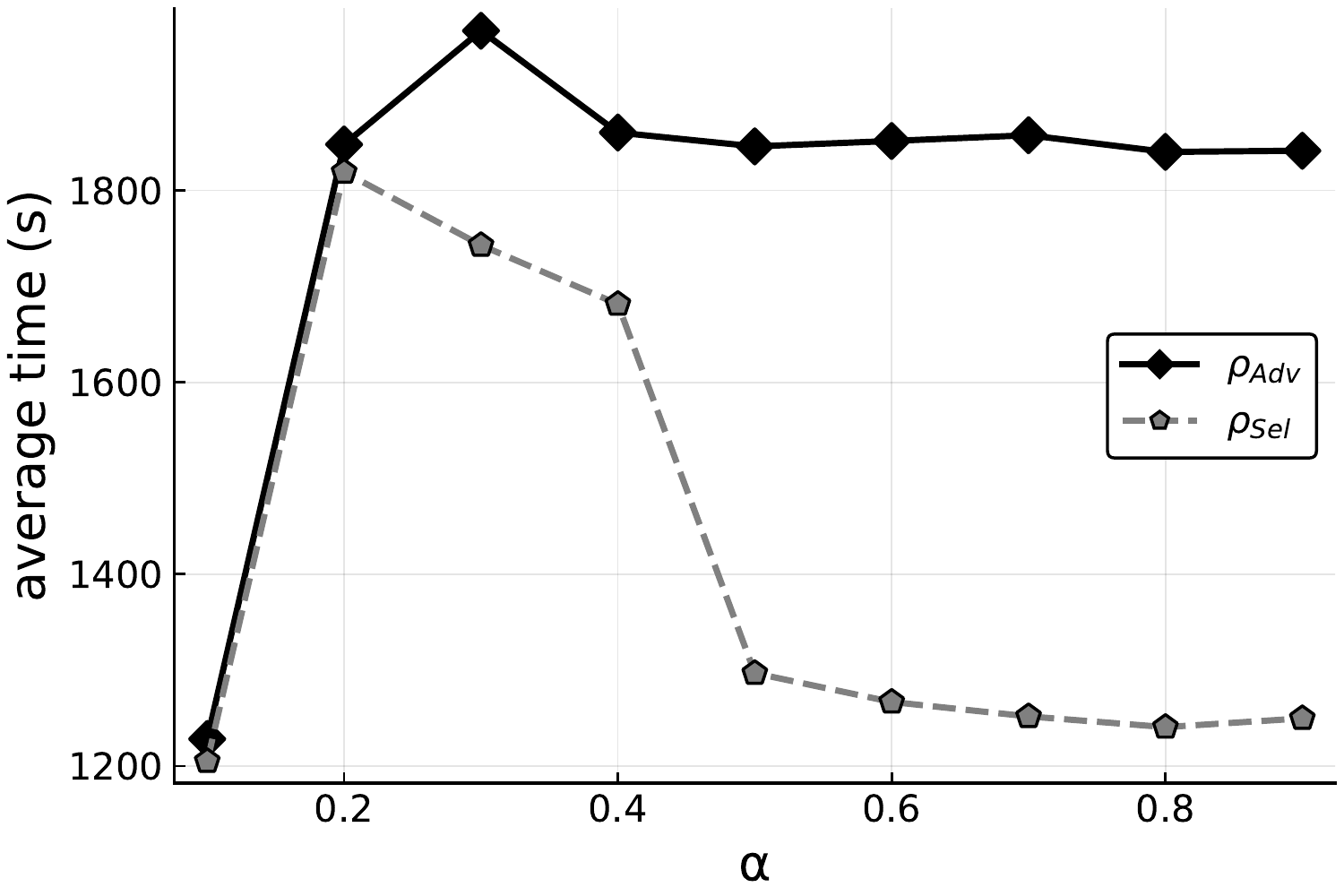}
	\caption{The average ratios $\rho_k$ and the average times for computing them for the knapsack problem with $n=400$. The time required to compute $\rho_H$ is negligible.} \label{fig10}
\end{figure}

In Figure~\ref{fig10} the results for $n=400$ are shown. One can observe similar relation between $\rho_{Sel}$ and $\rho_{Adv}$ as for the smaller problem with $n=100$.  However, the adversarial lower bound is now harder to compute and most instances were not solved to optimality (Algorithm~\ref{alg1} was terminated after the time of 600 seconds was exceeded). Observe that there is no significant improvement of $\rho_{Adv}$ over $\rho_{h}$.  

\subsection{Summary of the tests}

Let us briefly summarize the results of the tests. For the assumed method of data generation, the ratio $\rho(\pmb{c}_0)$ is almost always not greater than~2. Furthermore $\rho(\pmb{c}_0)$ can be computed efficiently for quite large instances, with thousands of variables. This suggests that the better of solutions $\underline{\pmb{x}}$, $\overline{\pmb{x}}$ has, for the tested instances, the empirical approximation ratio less than~2. One can conclude that this ratio is indeed significantly smaller than~2, by using better lower bounds. However, computing these lower bounds is more time consuming and can be done efficiently for smaller instances.

The solutions $\underline{\pmb{x}}$, $\overline{\pmb{x}}$  can be computed by solving the recoverable problem. For the assignment and knapsack problem, the recoverable problem is not particularly difficult to solve by \texttt{CPLEX}. However, the evaluation problem is more difficult, as we have to use the relaxation algorithm to perform this task. For large instances, we should assume more time for executing Algorithm~\ref{alg1}. We can then choose the solution among $\underline{\pmb{x}}$, $\overline{\pmb{x}}$, which has better upper bound on the value of $\textsc{Eval}(\pmb{x})$.

Notice that the techniques proposed in this paper are  general. For specific problem $\mathcal{P}$, the incremental, recoverable and evaluation problems can be solved by specialized algorithms (even in polynomial time). So, one can obtain better estimations for larger instances.

\section{Conclusions}

In this paper we considered a general class of 0-1 optimization problems, which can be polynomially solvable or NP-hard. The recoverable robustness concept was applied to take into account the possibility of performing a recourse action on the current first-stage solution. We proposed to use a polyhedral uncertainty representation. This model of uncertainty, containing the continuous interval budgeted uncertainty as a special case, can be easy to provide in practical applications. Moreover, it may lead to more tractable problems than other uncertainty representations. Unfortunately, the resulting min-max-min problem can be still complex to solve. Instead of solving the problem to optimality, we proposed to use some approximate solutions. The quality of these solutions can be estimated by using various lower bounds. One can apply this approach to quite large instances of the recoverable version of any 0-1 programming problem. In this paper we did not consider any particular 0-1 optimization problem. The computations can be done more efficiently if a polynomial algorithm is known for the incremental or recoverable problem. However, this is the case only for very specific problems such as selection or minimum spanning tree.

\subsubsection*{Acknowledgements}
This work was  supported by
 the National Science Centre, Poland, grant 2017/25/B/ST6/00486.


\end{document}